\documentclass[letterpaper]{sig-alternate-2013}
\usepackage{booktabs} % For formal tables
\usepackage{color}
\usepackage{graphicx}
\usepackage{balance}  % for  \balance command ON LAST PAGE  (only there!)
\usepackage{amsmath}
\usepackage{txfonts}
\usepackage{amssymb}
\usepackage{indentfirst}
\usepackage{times}
\usepackage{graphicx}
\usepackage{epsfig,tabularx,amssymb,amsmath,subfigure,multirow}
\usepackage[linesnumbered,ruled,noend]{algorithm2e}
\usepackage[noend]{algorithmic}
\usepackage{multirow}
\usepackage{graphicx,floatrow}
\usepackage{listings}
\usepackage{threeparttable}
\usepackage{tikz}
\usepackage[T1]{fontenc}
\usepackage{pgfplots}
\usepackage{paralist}
\usetikzlibrary{shapes.geometric, arrows}
\usepackage{float}
\usepackage{subfigure}
\usepackage{amsmath}
\usepackage{array}
\usepackage{url}
\usepackage[ruled]{algorithm2e}

\newtheorem{thm}{Theorem}
\newtheorem{defn}[thm]{Definition}
\newtheorem{prop}[thm]{Proposition}

\newcommand{\nop}[1]{}

\usepackage{floatrow}
\newfloatcommand{capbtabbox}{table}[][\FBwidth]

\begin{document}

\title{gSMat: A Scalable Sparse Matrix-based Join \\for SPARQL Query Processing}
\numberofauthors{8}
\author{
	\alignauthor
	Xiaowang Zhang\\
	\affaddr{School of Computer Science and Technology}\\
	\affaddr{Tianjin University}\\
	\affaddr{Tianjin, China}\\
	%\email{xiaowangzhang@tju.edu.cn}
	\alignauthor
	Mingyue Zhang\\
	\affaddr{School of Computer Science and Technology}\\
	\affaddr{Tianjin University}\\
	\affaddr{Tianjin, China}\\
	% \email{zhangmingyue@tju.edu.cn}
	\alignauthor
	Peng Peng\\
	\affaddr{College of Computer Science and Electronic Engineering}\\
	\affaddr{Hunan University}\\
	\affaddr{Changsha, China}\\
	\and
	\alignauthor Jiaming Song\\
	\affaddr{School of Computer Science and Technology}\\
	\affaddr{Tianjin University}\\
	\affaddr{Tianjin, China}\\
	% 5th. author
	\alignauthor Zhiyong Feng\\
	\affaddr{School of Computer Software}\\
	\affaddr{Tianjin University}\\
	\affaddr{Tianjin, China}\\
	%      \email{fogartys@amesres.org}
	% 6th. author
	\alignauthor Lei Zou\\
	\affaddr{Institute of Computer Science and Technology}\\
	\affaddr{Peking University}\\
	\affaddr{Beijing, China}\\
	%  \email{cpalmer@prl.com}
}

\maketitle
\begin{abstract}
Resource Description Framework (RDF) has been widely used to represent \emph{information} on the web, while SPARQL is a standard query language to manipulate RDF data. Given a SPARQL query, there often exist many joins which are the bottlenecks of efficiency of query processing. Besides, the real RDF datasets often reveal strong \emph{data sparsity}, which indicates that a resource often only relates to a few resources even the number of total resources is large. In this paper, we propose a sparse matrix-based (SM-based) SPARQL query processing approach over RDF datasets which considers both join optimization and data sparsity.
Firstly, we present a SM-based storage for RDF datasets to lift the storage efficiency, where valid edges are stored only, and then introduce a predicate-based hash index on the storage.
Secondly, we develop a scalable SM-based join algorithm for SPARQL query processing.
Finally, we analyze the overall cost by accumulating all intermediate results and design a query plan generated algorithm.
Besides, we extend our SM-based join algorithm on GPU for parallelizing SPARQL query processing.
We have evaluated our approach compared with the state-of-the-art RDF engines over benchmark RDF datasets and the experimental results show that our proposal can significantly improve SPARQL query processing with high scalability.
\end{abstract}
\section{Introduction}\label{sec:intro}
Resource Description Framework (RDF)~\cite{RDF} is a popular data model for information on the Web of the form a triple: (\textit{subject}, \textit{predicate}, \textit{object}). An RDF dataset can also be described as a directed labeled graph, where subjects and objects are vertices and triples are edges with labels (predicates). SPARQL~\cite{sparql1.0}, as the standard query language for RDF data, is officially recommended by W3C in 2008 and its latest version SPARQL~1.1~\cite{sparql1.1} is recommended in 2013.

There are many existing approaches to evaluate SPARQL queries over RDF data, which can be roughly classified into two kinds of storage strategies: relation-based storing and graph-based storing. The former stores an RDF triple as a tuple in a ternary relation (e.g., RDF3X~\cite{rdf3x1}, Hexastore~\cite{Hexastore}, Jena2~\cite{Jena2}, and BitMat~\cite{BitMat}) and the latter stores RDF data as a directed labeled graph (e.g., gStore~\cite{gStore1}). While existing approaches improve the performance of query evaluation to a certain extent. However, their improvements are still limited in processing large-scale RDF data from a real world such as DBpedia and YAGO due to little consideration of an important feature called `` sparsity '' of those practical RDF data.

The sparsity of RDF data means that the neighbors of each vertex in an RDF graph take a quite small proportion of the whole vertices. In fact, the sparsity of RDF data exists everywhere. For instance, there are over 99.41\% nodes with at most 43 degrees (sum of out-degrees and in-degrees) in DBpedia (42966066 nodes in total, see Figure \ref{fig:DBpediaSparsityExp}) and over 95.17\% nodes with at most 39 degrees in YAGO (38734252 nodes in total, see Figure \ref{fig:YAGOSparsityExp}).

\begin{figure}[H]
	%**********************DBpedia
	\subfigure[DBpedia Dataset]{
		\begin{minipage}[h]{0.45\linewidth}
			\centering
			\scalebox{0.45}{
				\begin{tikzpicture}
				\begin{axis}[
				symbolic x coords={43,430,4297,42967,429661,4296607,42966066},
				enlarge x limits=0.15,
				xtick = data,
				ybar,
				ymin=0,ymax=100,
				bar width=13pt,
				ylabel={The Percentage of Subjects (percent)},
				xlabel={The Number of Relations},
				ylabel near ticks,
				xlabel near ticks,
				x tick label style={rotate=45,anchor=east},
				nodes near coords,
				nodes near coords align={rotate=45,anchor=west},
				]
				
				\addplot+
				[color=cyan, fill=cyan]
				coordinates {
					(43,99.4082) (430,0.5402) (4297,0.0435) (42967,0.0078) (429661,0.0002) (4296607,0) (42966066,0)
				};
				
				\end{axis}
				\end{tikzpicture}
			}
		\end{minipage}
		\label{fig:DBpediaSparsityExp}%
	}
	%**********************YAGO
	\subfigure[YAGO Dataset]{
		\begin{minipage}[h]{0.45\linewidth}
			\centering
			\scalebox{0.45}{
				\begin{tikzpicture}
				\begin{axis}[
				symbolic x coords={0,39,388,3874,38735,387343,3873426,38734252},
				enlarge x limits=0.15,
				xtick = data,
				ybar,
				ymin=0,ymax=100,
				bar width=13pt,
				ylabel={The Percentage of Subjects (percent)},
				xlabel={The Number of Relations},
				ylabel near ticks,
				xlabel near ticks,
				x tick label style={rotate=45,anchor=east},
				nodes near coords,
				nodes near coords align={rotate=45,anchor=west},
				]
				
				\addplot
				[color=cyan, fill=cyan]
				coordinates {
					(39,95.1664) (388,4.6248) (3874,0.2010) (38735,0.0074) (387343,0.0004) (3873426,0) (38734252,0)
				};
				\end{axis}
				\end{tikzpicture}
			}
		\end{minipage}
		\label{fig:YAGOSparsityExp}%
	}
	\vspace*{-0.5cm}
	\caption{Data Sparsity in YAGO and DBpedia datasets}
	\label{fig:SparsityExp}%
\end{figure}
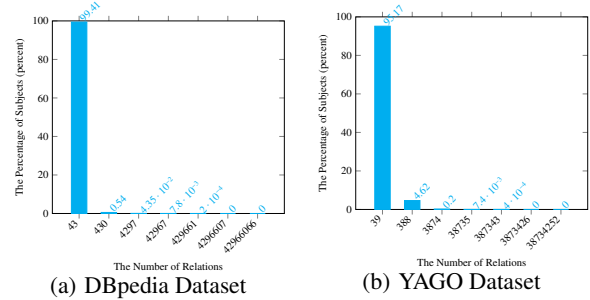

SPARQL is built on basic graph pattern (BGP) and SPARQL algebra operators. The join for concatenating variables is the core operation of SPARQL query evaluation since BGP is the join of triple patterns. In an RDF graph, the join (concatenation) of two vertices can be computed by the product of their adjacent matrices. For instance, BitMat~\cite{BitMat} employs matrices to compute join of SPARQL over RDF data. This matrix-based computation of join can be completely parallelized since each element is computed separately in the product of matrices. Taking this hyper-parallelized computing method, matrix-based join can efficiently support query evaluation over large-scale data better. Besides, the matrix-based join is friendly to support new parallelizing computing units such as Graphics Processing Units (GPUs).

In this paper, we propose a sparse matrix-based approach (named gSMat) to improve SPARQL join query evaluation over large-scale RDF data. The RDF data represented by graph is classified according to the relationships of edges, and then the data represented by each relation is stored in the form of a sparse matrix. We transform the problem of SPARQL query processing to the multiplication operations of multiple sparse matrices, and each triple pattern corresponds to a sparse matrix in the SPARQL query. A \emph{sparse matrix} is a kind of matrix which consists of large numbers of zero values with very few dispersed non-zero ones. Comparing to general matrices, sparse matrices characterize the sparsity of RDF data well.
The experimental results show that our proposal can significantly improve BGP query evaluation, especially over practical RDF data with 500 million triples.
For instance, over the benchmark datasets and queries, our gSMat has speedup up to 114 w.r.t. gStore and speedup up to 11.3 w.r.t. RDF-3X. Moreover, gSMat with GPU has speedup up to 128.8 w.r.t. gStore and speedup up to 16.13 w.r.t. RDF-3X. And over the real datasets and queries, our gSMat has speedup up to 46.4 w.r.t. gStore and speedup up to 33.6 w.r.t. RDF-3X. Moreover, gSMat with GPU has speedup up to 78.9 w.r.t. gStore and speedup up to 53.0 w.r.t. RDF-3X. Specifically, our major contributions are summarized in the followings:
\begin{itemize}
	\item We propose a sparse matrix-based storage for exactly characterizing the sparsity of RDF data, where a predicate corresponds to a sparse matrix. The sparse matrix-based storage provides a compact and efficient method to store RDF data physically and also supports high-performance algorithm of query execution.
	\item We design a query plan generated algorithm and analyze the overall cost by accumulating all intermediate results during runtime.
	It is based on the statistics of sparse matrices to reduce the size (number) of intermediate results so that we can generate an optimal execution plan for a join query.
	\item We develop a sparse matrix-based SPARQL join algorithm and its GPU-based parallel extended algorithm. Moreover, we optimize our GPU-based algorithm in some aspects such as data transferring, thread scheduling, shared and global memory management.
\end{itemize}

The rest of this paper is organized as follows: the next section introduces the related works. Section~\ref{sec:Preliminaries} briefly introduces RDF, SPARQL, and GPU computing. Section~\ref{sec:overview} presents the overview of framework of our proposal. Section~\ref{sec:storage} presents the sparse matrix-based storage and Section~\ref{sec:plan} presents the sparse matrix-based processing, respectively. Section~\ref{sec:GPU} introduces the extension of our proposal on GPU. Section~\ref{sec:ex} discusses experiments and evaluations. Finally, Section \ref{sec:conclusion} summarizes this paper.

\section{Related Work}\label{sec:related}
In this section, we review recent single-machine RDF databases, which we believe are most related to gSMat, and summarize the main differences to our engine. And we discuss some of the work done on the GPU query.

\subsection{SPARQL Query Processing}
Existing approaches to evaluate SPARQL queries over
RDF data can be classified into two categories:
relation-based systems and graph-based systems.

The relation-based systems apply a relational approach to store and index RDF data. Several systems, such as Jena~\cite{Jena1,Jena2}, Oracle~\cite{Oracle}, Sesame~\cite{Sesame2}, 3store~\cite{3store} and SOR~\cite{SOR} use a giant table to maintain triples, with three columns corresponding to subject, predicate and object, respectively. A SPARQL query will be transformed into a SQL query, and evaluated through multiple self-joins of the table. However, a significant amount of self-joins in relational database result in a long answering time, which is a potential bottleneck of the systems.

Additionally, several relation-based systems such as Hexastore~\cite{Hexastore}, RDF-3x~\cite{rdf3x1,rdf3x2,rdf3x3}, BitMat~\cite{BitMat} and TripleBit~\cite{triplebit}, employ specialized optimization techniques based on the features of RDF data and SPARQL queries. Hexastore~\cite{Hexastore} and RDF-3x~\cite{rdf3x1,rdf3x2,rdf3x3} build a set of indices that cover all possible permutations of S, P and O, in order to speed up the joins. TripleBit~\cite{triplebit} uses a two-dimension matrix to represent RDF triples, with subjects and objects as row and predicates as column. Then, it uses '1' to label the relation, otherwise uses '0'. Thus, there are only two '1' in each column, which are easy to be recorded. Furthermore, the triple matrix is divided into submatrices by the same properties and stored by column. Moreover, BitMat~\cite{BitMat} numbers every elements of RDF triples and builds bitmap indexes to collect the candidates of queries. Similar to property table, TripleBit and BitMat suffer from the large waste of space. Instead, we save the storage space and facilitate the processing in a way that represents the data as a sparse matrix and only maintains the actual relations in RDF graph.

Recently, a number of graph-based approaches were proposed to store RDF triples in graph models, such as gStore~\cite{gStore1,gStore2}, dipLODocus$_{[RDF]}$~\cite{dipLOD}, Turbo$_{HOM++}$ and AMBER~\cite{AMBER}. These graph-based approaches typically see SPARQL query processing as subgraph matching, which help reserve and query semantic information. gStore~\cite{gStore1,gStore2} maps all predicates and predicate values to binary bit strings which are then organized as a VS*-tree. Since every layer of VS*-tree is a summary of the whole RDF graph, gStore has capacity to process SPARQL query efficiently. dipLODocus$_{[RDF]}$ starts by a mixed storage considering both graph structure of RDF data and requirement of data analysis, in order to find molecule clusters and help accelerate queries through clustering related data.
Turbo$_{HOM++}$
~\cite{turbohom} develops TurboISO~\cite{turboiso} by transforming RDF graphs into normal data graphs while AMBER~\cite{AMBER} represents RDF data and SPARQL query as multigraphs.

However, all the above systems are computationally expensive in preprocessing due to their lack of parallelism. In addition, query methods with Time-consuming traversal add run time overhead. Instead, focused on the sparsity of real RDF datasets, our approach generates sparse matrices and allows subqueries to run parallel, thus has the crucial advantage that it can be adapted to large-scale datasets.

\subsection{Query processing based on GPU}
We now briefly survey the techniques that use GPUs to improve the performance of database operations.  Current database researches identify the computational power of GPUs as a way to increase the performance of database systems. Since GPU algorithms are not necessarily faster than their CPU counterparts, it is important to use the GPU only if it is beneficial for query processing.
S.Breß et al.~\cite{GPU1}~\cite{GPU2} extend CPU/GPU scheduling framework to support hybrid query processing in database systems. ~\cite{GPU4} focuses on accelerating SELECT queries and describes the considerations in an efficient GPU implementation.
~\cite{GPU5} accelerates the search in big RDF data by exploiting modern multi-core architectures based on GPU.
In~\cite{GPU3}, a new efficient and scalable index is proposed. These data structures have an edge over others in terms of their implementation as a parallel algorithm using the CUDA (Compute Unified Device Architecture) framework. ~\cite{GPU6} shows that usage of graphic card for intensive calculations over large knowledge bases in RDF format can be another way to decrease computational time.

We focus on GPU-based algorithms for the SPARQL join operation, which is a core operator in SPARQL query. Moreover, our algorithms are based on a multi-core SIMD(Single Instruction Multiple Data) architecture model of the GPU, and thus can be applied to CPUs of a similar architecture.

\section{Preliminaries} \label{sec:Preliminaries}
In this section, we briefly recall RDF and SPARQL in~\cite{RDF,RDF1.1,sparql1.0,sparql1.1}.

\subsection{RDF}
Let $U$ be an infinite countable set of constants. Let $U_s$, $U_{p}$, and,  $U_{o}$ be three subsets of $U$.  An RDF triple is of the form (\textit{subject}, \textit{predicate}, \textit{object}) (for short, $(s, p, o)$) where \textit{subject} $\in U_{s}$ denotes an entity or a class of resources; \textit{predicate} $\in U_{p}$ denotes attributes or relationships between entities or classes; and \textit{object} $\in U_{o}$ is an entity, a class of resources, or a literal value.  An \textit{RDF dataset} $D$ is a set of RDF triples.  An RDF dataset represents a labeled directed graph where \textit{predicate} is taken as a label~\cite{gutierrez_survey}, so also called \textit{RDF graph}. In this sense, we assume that $U_{p} \cap (U_{s} \cup U_{o}) = \emptyset$, that is, predicates are never entities or classes.

Let $N_s$, $N_p$, and $N_o$ denote the set of subjects, predicates, and objects, respectively. Let $D$ be an RDF dataset. Let $N_s(D)$, $N_p(D)$, and $N_o(D)$ denote the set of subjects, predicates, and objects occurring in $D$, respectively. Formally, we define $N_s(D)$, $N_p(D)$, and $N_o(D)$ as follows:
\begin{itemize}
	\item $N_s(D):= \{u \mid \exists\,v \in U_{p}, w \in U_{o}, ~(u, v, w) \in D\}$;
	\item $N_p(D):= \{v \mid \exists\,u \in U_{s}, w \in U_{o}, ~(u, v, w) \in D\}$;
	\item $N_o(D):= \{w \mid \exists\,u \in U_{s}, v \in U_{p}, ~(u, v, w) \in D\}$.
\end{itemize}

For instance, an RDF dataset $D$ is shown in the right side of Figure \ref{fig:rdf}, which is taken as a directed graph shown in the left side of Figure \ref{fig:rdf}.

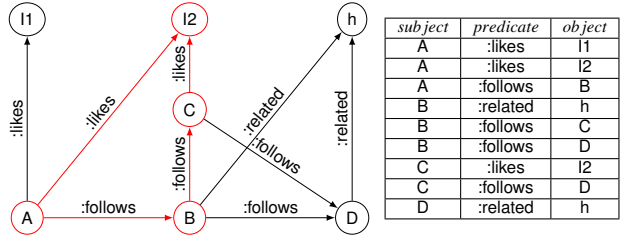
\begin{figure}[H]
	\begin{minipage}{0.4\linewidth}
		\small
		\scalebox{0.8}{
			\begin{tikzpicture}[
			x=0.3cm,
			y=0.3cm,
			font=\sffamily,
			v/.style={shape=circle,draw,inner sep=1pt,minimum size=2mm},
			to/.style={->,>=latex},
			p/.style={inner sep=2pt,above,sloped} ]
			%\draw[help lines,step=1] (0,0) grid (42,12);
			
			% nodes
			\node[shape=ellipse,text centered, draw=red] (v1) at (0,1) {A};
			\node[shape=ellipse,text centered, draw=red] (v2) at (9,1) {B};
			\node[shape=ellipse,text centered, draw=red] (v3) at (9,7) {C};
			\node[shape=ellipse,text centered, draw=black] (v4) at (18,1) {D};
			\node[shape=ellipse,text centered, draw=black] (v5) at (0,12) {I1};
			\node[shape=ellipse,text centered, draw=red] (v6) at (9,12) {I2};
			\node[shape=ellipse,text centered, draw=black] (v7) at (18,12) {h};
			% edges
			\draw[to,draw=red] (v1) to node[p,above] {:follows} (v2);
			\draw[to,draw=red] (v2) to node[p,above] {:follows} (v3);
			\draw[to] (v2) to node[p] {:follows} (v4);
			\draw[to] (v3) to node[p] {:follows} (v4);
			\draw[to] (v1) to node[p] {:likes} (v5);
			\draw[to,draw=red] (v1) to node[p] {:likes} (v6);
			\draw[to] (v2) to node[p] {:related} (v7);
			\draw[to,draw=red] (v3) to node[p] {:likes} (v6);
			\draw[to] (v4) to node[p] {:related} (v7);
			\end{tikzpicture}
		}
	\end{minipage}
	\hspace*{1.5cm}
	\begin{minipage}{0.4\linewidth}
		\small
		\scalebox{0.8}{
			\begin{tabular}{|c|c|c|}
				\hline
				\textbf{$subject$} &\textbf{$predicate$}& \textbf{$object$} \\ \hline
				\textsf{A} &\textsf{:likes}& \textsf{I1} \\\hline
				\textsf{A} &\textsf{:likes}& \textsf{I2} \\\hline
				\textsf{A} &\textsf{:follows}& \textsf{B} \\\hline
				\textsf{B} &\textsf{:related}& \textsf{h} \\\hline
				\textsf{B} &\textsf{:follows}& \textsf{C} \\\hline
				\textsf{B} &\textsf{:follows}& \textsf{D} \\\hline
				\textsf{C} &\textsf{:likes}& \textsf{I2} \\\hline
				\textsf{C} &\textsf{:follows}& \textsf{D} \\\hline
				\textsf{D} &\textsf{:related}& \textsf{h} \\\hline
			\end{tabular}
		}
	\end{minipage}
	\caption{An Example of an RDF dataset $D_G$}\label{fig:rdf}
\end{figure}

\subsection{SPARQL}%% Peng to be revised
SPARQL (named recursively, \emph{SPARQL Protocol and RDF Query Language}) is the official recommended standard for RDF query language by W3C and it defines the syntax and semantics of RDF query language \cite{SPARQLsemantics}.  SPARQL query language is based on triple patterns, so-called $tp =(s,p,o)$ and there may be variables in any position, such as: $\{?\textit{A},\, \textit{:follows},\, ?\textit{B}\}$. And a set of triple patterns forms a basic graph pattern (BGP).

\nop{
	SPARQL is a query language that retrieves information from RDF graph, which is built on previous RDF query languages (rdfDB, RDQL, and SeRQL) for accessing any data resource that can be mapped to the RDF model. A SPARQL query defines a graph pattern $P$ that is matched against an RDF graph $G$. This is done by replacing the variables in $P$ with elements of $G$ such that the resulting graph is contained in $G$ (pattern matching).
}

A common SPARQL query contains a group of BGP queries, whose conjunctive fragment allows to express the core form:
\[
\emph{Select}|\emph{Project}|\emph{Join}
\]
database queries. For instance, considering a query as follows: for the above SPARQL query, the SPARQL query is represented as a query graph (or query pattern) in Figure \ref{fig:query}.

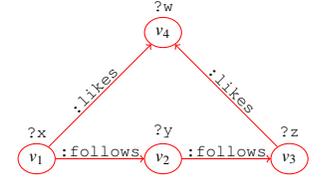
\begin{figure}[H]
	\begin{minipage}{0.4\linewidth}
		\footnotesize
		\scalebox{1}{
			\fbox{
				\parbox{3.5cm}{
					\textbf{SELECT} ?x ?y ?z ?w\\
					\textbf{WHERE} \{  \\
					\hspace*{0.4in} ?x :follows ?y~.\\
					\hspace*{0.4in} ?y :follows ?z~.\\
					\hspace*{0.4in} ?x :likes ?w~.\\
					\hspace*{0.4in} ?z :likes ?w~.\}\\
				}
			}
		}
	\end{minipage}
	\hspace*{1cm}
	\begin{minipage}{0.4\linewidth}
		\small
		\scalebox{0.8}{
			\begin{tikzpicture}[
			%node distance=100mm,
			x=0.35cm,
			y=0.35cm,
			font=\sffamily,
			v/.style={shape=circle,draw,inner sep=1pt,minimum size=2mm},
			to/.style={->,>=latex},
			p/.style={inner sep=30pt,above,sloped} ]
			\node[shape=ellipse,text centered, draw=red] [label=above:\texttt{?x}] (v1) at (0,6){$v_1$} ;
			\node[shape=ellipse,text centered, draw=red] [label=above:\texttt{?y}] (v2) at (6,6) {$v_2$};
			\node[shape=ellipse,text centered, draw=red] [label=above:\texttt{?z}] (v3) at (12,6) {$v_3$};
			\node[shape=ellipse,text centered, draw=red] [label=above:\texttt{?w}] (v4) at (6,12) {$v_4$};
			
			\draw[->,draw=red](v1) to node[p,inner sep=1pt] {\texttt{:follows}} (v2);
			\draw[->,draw=red](v2) to node[p,inner sep=1pt]{\texttt{:follows}} (v3);
			\draw[->,draw=red](v1) to node[p,inner sep=1pt] {\texttt{:likes}} (v4);
			\draw[->,draw=red](v3) to node[p,inner sep=1pt] {\texttt{:likes}} (v4);
			\end{tikzpicture}
		}
	\end{minipage}
	\caption{An example of SPARQL query}\label{fig:query}
\end{figure}

Given a SPARQL query,
we can parse it into a query graph (structure) where each triple pattern is taken as an edge labelled by its predicate and two edges are connected by their common constant or variable. For example, we illustrate a query $Q_g$ consisting of nine triple patterns (a benchmark query in WatDiv~\cite{WatDiv}) as shown in Figure~\ref{fig:exam-Q-g} and the query graph structure of $Q_g$ after parsing is shown in the left side of Figure~\ref{fig:querygraph}.
\begin{figure}[h]
	{\small \begin{center}
			\fbox{
				\parbox{8cm}{
					\textbf{SELECT} ?v0 ?v1 ?v2 ?v4 ?v5 ?v6 ?v7 ?v8\\
					\textbf{WHERE} \{\\
					\hspace*{0.1in} ?v0 <http://xmlns.com/foaf/homepage> ?v1~.\\
					\hspace*{0.1in} ?v2	<http://purl.org/goodrelations/includes> ?v0~.\\
					\hspace*{0.1in} ?v0	<http://ogp.me/ns\#tag>	?v3~.\\
					\hspace*{0.1in} ?v0	<http://schema.org/description>	?v4~.\\
					\hspace*{0.1in} ?v0	<http://schema.org/contentSize>	?v8~.\\
					\hspace*{0.1in} ?v1	<http://schema.org/url>	?v5~.\\
					\hspace*{0.1in} ?v1	<http://db.uwaterloo.ca/\#galuc/wsdbm/hits>	?v6~.\\
					\hspace*{0.1in} ?v1	<http://schema.org/language>	?v9~.\\
					\hspace*{0.1in} ?v7	<http://db.uwaterloo.ca/\#galuc/wsdbm/likes>	?v0~.\\
					\hspace*{0.1in}\}
				}
			}
	\end{center}}
	\caption{A running example of query $Q_{g}$}\label{fig:exam-Q-g}
\end{figure}

\begin{figure}[H]
	\subfigure
	{
		\begin{minipage}[htbp]{1\linewidth}
			\centering
			\scalebox{1.1}{
				\includegraphics[width=8cm]{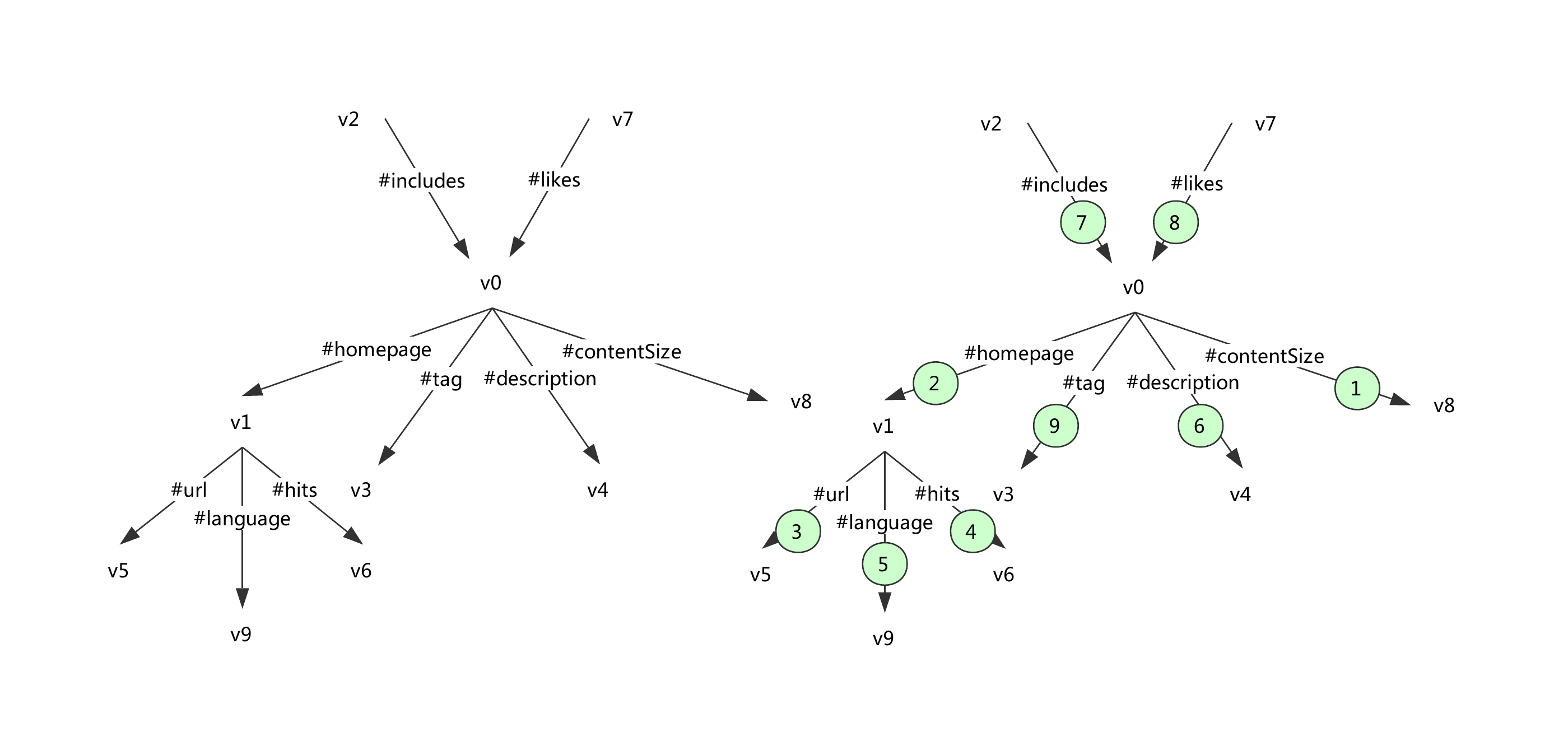}
			}
		\end{minipage}
	}
	\caption{The query graph structure of $Q_g$} \label{fig:querygraph}
\end{figure}

In addition, each triple pattern can be mapped to a sparse matrix. According to the predicate of each triple pattern, that is, the edge relationship, we can index sparse matrix representing this relationship.
A SPARQL query processing problem is about join operation of multiple triple patterns, and now it is the join operation of multiple sparse matrices that can be transformed into a basic multiplication operation of sparse matrices.
In other words, we transfrom the SPARQL query problem into a sparse matrix-based multiplication.

\section{Overview of \lowercase{g}SM\lowercase{at}}\label{sec:overview}
In this section, we give an overview of gSMat (i.e., \textit{g}raph \textit{S}parse \textit{Mat}rix-based SPARQL query processing). The framework of gSMat is shown in Figure~\ref{fig:architecture}.
\begin{figure}[H]
	\subfigure
	{
		\begin{minipage}[htbp]{1\linewidth}
			\centering
			\scalebox{1}{
				\includegraphics[width=8cm]{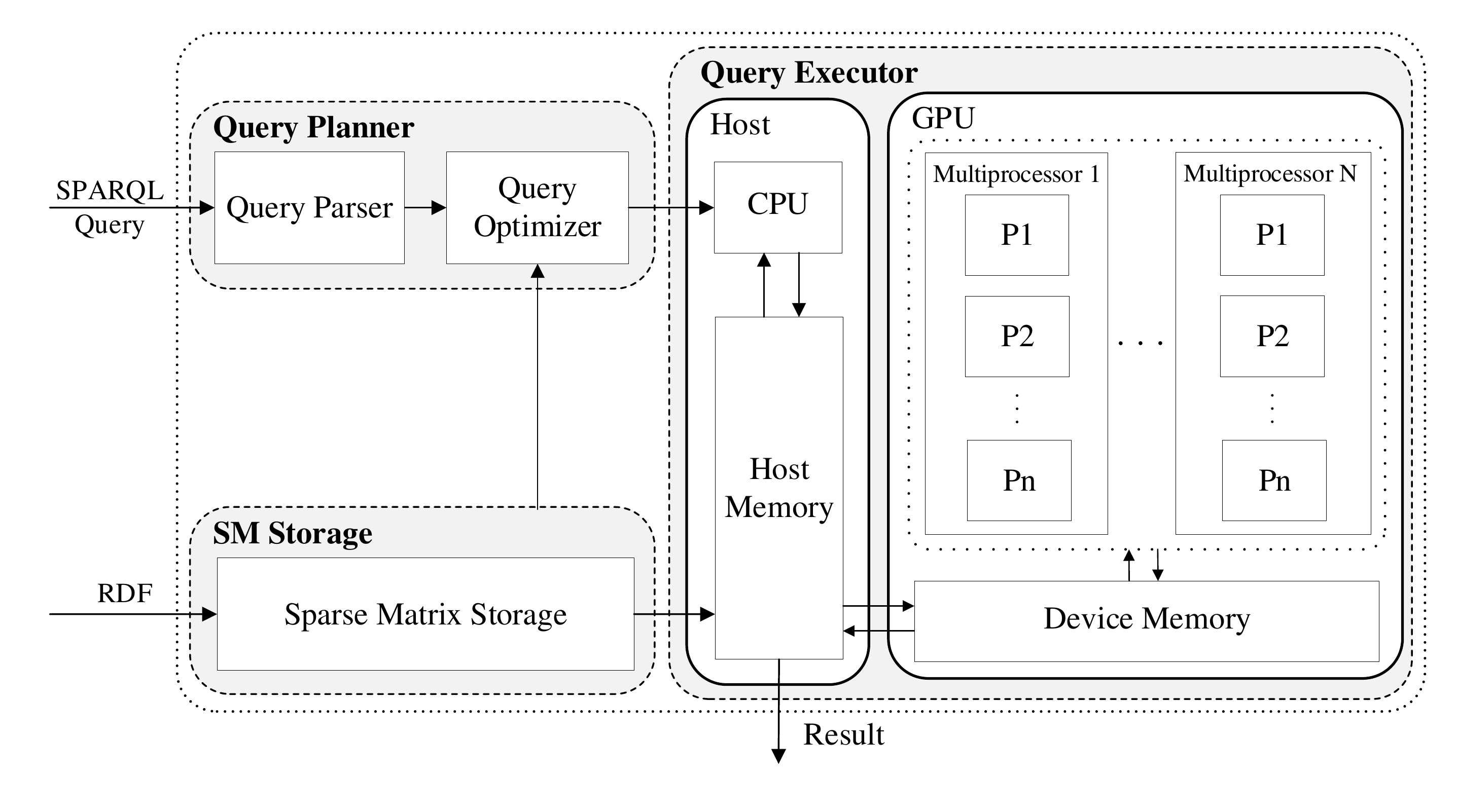}
			}
		\end{minipage}
	}
	\caption{The framework of gSMat} \label{fig:architecture}
\end{figure}

\begin{figure*}[htbp]
	\centering
	\includegraphics[width=1\textwidth]{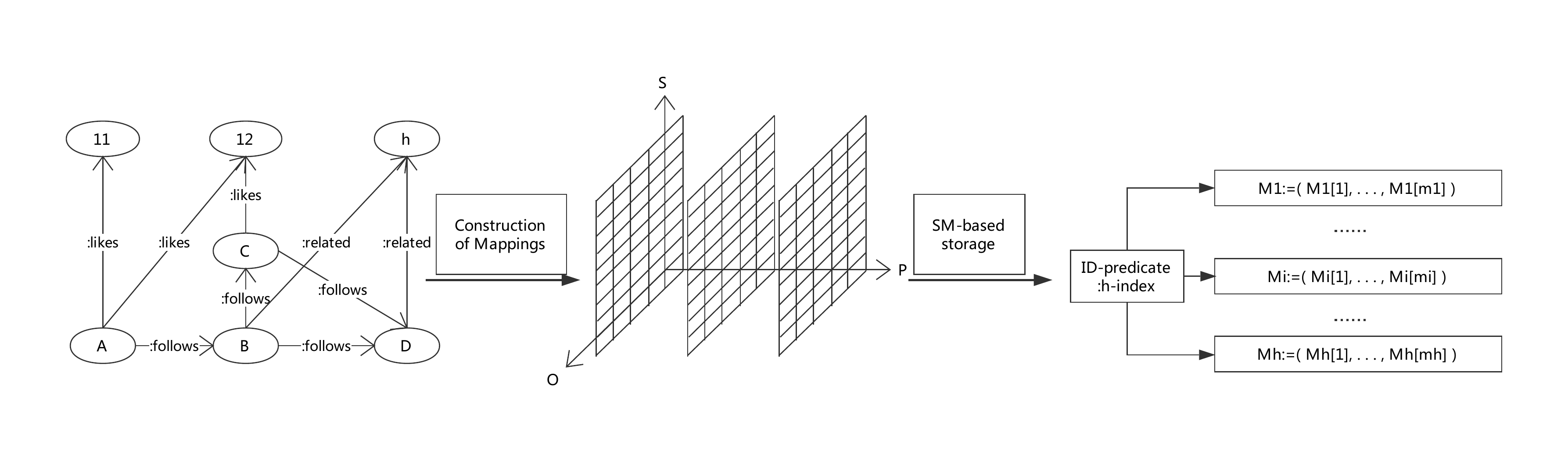}
	\caption{The procedure of SM-based storage of an RDF dataset}\label{fig:smstore}
\end{figure*}

The whole framework of gSMat contains three modules, namely, \emph{SM Storage}, \emph{Query Planner}, and \emph{Query Executor}.
\begin{description}
	\item[SM Storage] This module maintains a collection of SM-based tables indexed by predicates and the statistics of predicates from RDF dataset. The sparse matrix-based storage provides a compact and efficient method to store RDF data physically and also supports high-performance algorithm of query execution (discussed in Section~\ref{sec:storage}).
	\item[Query Planner] This module contains two parts, \emph{Query Parser} and \emph{Query Optimizer}. The former parses a SPARQL query to a query graph and the latter generates the optimal query plan based on the statistics in the module of SM Storage. Here, in order to achieve the optimal query plan, we analyze the overall cost by accumulating all intermediate results during the whole of relations joining to evaluate the cost of the query (discussed in Section~\ref{sec:planer}).
	\item[Query Executor] This module contains two parts, namely, CPU computation and GPU computation. SPARQL queries can be evaluated in this module by applying two different strategies (CPU-based and GPU-based) (discussed in Section~\ref{sec:exe} and Section~\ref{sec:GPU}). Note that this module provides a switch of GPU so that we could conveniently use GPU to process queries during the query execution.
\end{description}

In the following three sections (Section~\ref{sec:storage}, Section~\ref{sec:plan}, and Section~\ref{sec:GPU}), we introduce the three modules of gSMat in details.

\section{Sparse Matrix-based Storage}\label{sec:storage}
In this section, we present SM-based storage.

\subsection{RDF Cube}
Firstly, we define the notion of \emph{RDF Cube} as a logical model of SM-based storage for RDF datasets.

\begin{defn}[Cube]\label{def:Cube}
	Let $n$ and $m$ be two positive integers. Let $N = \{1,\ldots, n\}$ and $M = \{1,\ldots,m\}$. A ternary relation $\mathcal{C} \subseteq N \times M \times M$ is an $n$-index Cube if the followings hold:
	\begin{itemize}
		\item $N = \{i \mid \exists\,j,k, s.t. (i, j, k) \in \mathcal{C}\}$;
		\item $M =\{j,k \mid \exists\,i, s.t. (i, j, k) \in \mathcal{C}\}$.
	\end{itemize}
\end{defn}

Next, we will model RDF datasets as cubes. To do so, we firstly introduce a notion of \emph{RDF Cube}.

\begin{defn}[RDF Cube]
	Let $D$ be an RDF dataset. An RDF Cube $\mathcal{C}_{D}$ is a $|N_p(D)|$-index Cube from $h_p(N_p(D))$ to $h_{so}(N_s(D) \cup N_o(D))$ defined as follows:
	\[
	\mathcal{C}_{D}:= \{(h_p(v), h_{so}(u), h_{so}(w)) \mid (u, v, w) \in D\}
	\]
	where
	\begin{itemize}
		\item $h_p$ is a bijective mapping from $N_p(D)$ to $\{1,\ldots, |N_p(D)|\}$; \item $h_{so}$ is a bijective mapping from $N_s(D) \cup N_o(D)$ to $\{1,\ldots, |N_s(D) \cup N_o(D)|\}$.
	\end{itemize}
\end{defn}

Note that $N_s(D)$, $N_p(D)$, and $N_o(D)$ are collections of all subjects, predicates, and objects occurring in $D$, respectively.

Based on RDF Cube, we present a sparse matrix-based storage in the following three steps:
\begin{itemize}
	\item Constructing two bijective mappings $h_P$ and $h_{SO}$ from the set of predicates and the union of sets of subjects and objects;
	\item Generating RDF Cube of RDF datasets;
	\item Storing RDF Cube as sets of sparse matrix via index of predicates. And recording the statistics of each sparse matrix at the same time.
\end{itemize}

Next, we introduce how to design a sparse matrix-based storage of an RDF dataset.

\subsection{Generation of RDF Cubes}
In this step, we encode strings of raw data into postive integers so that we can construct two bijective mappings $h_p$ and $h_{so}$. The goal of encoding is to reduce space usage and disk {I/O}. There are many methods to encode, such as Hash.

Given an RDF dataset $D$, the encoding process encodes all subjects, predicates and objects by the order of occurrance of triples in $D$. Finally, this process terminates till all subjects, predicates, and objects encoded completely.

For instance, consider the RDF dataset $D_G$ in Figure~\ref{fig:rdf}, we have
\begin{itemize}
	\item $N_{s}(D_G) = \{\textsf{A}, \textsf{B}, \textsf{C}, \textsf{D}\}$;
	\item $N_{p}(D_G) = \{:\textsf{likes}, :\textsf{follows}, :\textsf{related}\}$;
	\item $N_{o}(D_G) = \{\textsf{l1}, \textsf{l2}, \textsf{B}, \textsf{h}, \textsf{C}, \textsf{D}\}$.
\end{itemize}

We construct $h_p$ and $h_{so}$ shown in Figure~\ref{fig:hash-rdf}.
\begin{figure}[H]
	\begin{minipage}{0.45\linewidth}
		\centering
		\begin{tabular}{|c|c|}
			\hline
			\textsf{:likes} & 1 \\\hline
			\textsf{:follows} & 2 \\\hline
			\textsf{:related} & 3 \\\hline
		\end{tabular}
	\end{minipage}
	\hspace*{-1cm}
	\begin{minipage}{0.45\linewidth}
		\centering
		\begin{tabular}{|c|c|}
			\hline
			\textsf{A} & 1 \\\hline
			\textsf{B} & 4 \\\hline
			\textsf{C} & 6 \\\hline
			\textsf{D} & 7 \\\hline
			\textsf{l1} & 2 \\\hline
			\textsf{l2} & 3\\\hline
			\textsf{h} & 5\\\hline
		\end{tabular}
	\end{minipage}
	\caption{The construction of $h_p$ and $h_{so}$ for $D_G$}\label{fig:hash-rdf}
\end{figure}

Based on mappings $h_p$ and $h_{so}$ constructed in the previous step, we generate an RDF Cube of an RDF dataset, that is, a set of triples of postive integers (as codes).

For instance, consider the RDF dataset $D_G$ in Figure~\ref{fig:rdf}, we construct $\mathcal{C}_{D_G}$ shown in Table~\ref{tab:cube-rdf}.
\begin{table}[H]
	\caption{The RDF Cube $\mathcal{C}_{D_G}$ of an RDF dataset $D_G$}\label{tab:cube-rdf}
	\centering
	\begin{tabular}{|c|c|}
		\hline
		$D_G$ & $\mathcal{C}_{D_G}$ \\ \hline
		(\textsf{A}, \textsf{:likes}, \textsf{I1}) & (1, 1, 2) \\\hline
		(\textsf{A}, \textsf{:likes}, \textsf{I2}) & (1, 1, 3) \\\hline
		(\textsf{A}, \textsf{:follows}, \textsf{B}) & (1, 2, 4) \\\hline
		(\textsf{B}, \textsf{:related}, \textsf{h}) & (4, 3, 5) \\\hline
		(\textsf{B}, \textsf{:follows}, \textsf{C}) & (4, 2, 6) \\\hline
		(\textsf{B}, \textsf{:follows}, \textsf{D}) & (4, 2, 7) \\\hline
		(\textsf{C}, \textsf{:likes}, \textsf{I2}) & (6, 1, 3) \\\hline
		(\textsf{C}, \textsf{:follows}, \textsf{D}) & (6, 2, 7)\\\hline
		(\textsf{D}, \textsf{:related}, \textsf{h}) & (7, 3, 5) \\\hline
	\end{tabular}
\end{table}

\subsection{SM-based Storage of RDF Datasets}
In this step, we store RDF Cube as sets of sparse matrices via index of predicates. Let $h$ be a postive integer.  An $h$-index \emph{spare matrix} $M_{h}$ is an $m (statistics) \times 2$ matrix whose two columns are lalelled as $s$ and $o$, respectively and the matrix is labeled by $h$.

Let $\mathcal{C}$ be an RDF Cube and $h$ be a postive integer.  We construct an $h$-index spare matrix $M_h$ as follows:
\[
M_h:= (M_h[1], \ldots, M_h[m]);
\]
where for each $i \in \{1,\ldots, m\}$, the row vector $M_h[i]$ is defined as follows: $M_h[i]:= (j_i, k_i)$ iff $(j_i, h, k_i) \in \mathcal{C}$.

For instance, consider the RDF dataset $D_G$ in Figure~\ref{fig:rdf}, we construct three spare matrixes $M_1$, $M_2$, and $M_3$ shown in Figure~\ref{fig:SM-rdf}.
\begin{figure}[H]
	\begin{minipage}{0.30\linewidth}
		\centering
		\begin{tabular}{|c|c|}
			\hline
			1 & 2 \\\hline
			1 & 3 \\\hline
			6 & 3 \\\hline
		\end{tabular}
	\end{minipage}
	\hspace*{-1cm}
	\begin{minipage}{0.30\linewidth}
		\centering
		\begin{tabular}{|c|c|}
			\hline
			1 & 4 \\\hline
			4 & 6 \\\hline
			4 & 7 \\\hline
			6 & 7 \\\hline
		\end{tabular}
	\end{minipage}
	\hspace*{-1cm}
	\begin{minipage}{0.30\linewidth}
		\centering
		\begin{tabular}{|c|c|}
			\hline
			4 & 5 \\\hline
			5 & 7 \\\hline
		\end{tabular}
	\end{minipage}
	\caption{The three spare matrixes $M_1$, $M_2$, and $M_3$ for $D_G$}\label{fig:SM-rdf}
\end{figure}

From Figure~\ref{fig:smstore}, gSMat doesn't load triples of the form (\emph{x, y, z}), but (\emph{x, z}) pairs. So we only need 9 pairs, that is 18 units (in total, $M_1$, $M_2$, and $M_3$) to store the RDF dataset $D_{G}$ via spare matrix while we need 147 (7$\ast$7$\ast$3) units via adjacent matrix. So the sparse matrices could actually store RDF datasets in a compact way.

In a short, we can summarize the following advanages of SM-based storage:
\begin{itemize}
	\item {\em High-density storage capacity}: SM-based storage is based on sparse matrices where only non-zero elements are stored, that is to say, SM-based storage is an edge-based storage.
	\item {\em High-performance join computation}:  SM-based storage provides a high-performance join which is a multiplication of sparse matrices since multiplication of sparse matrices can mainly concern available elements in sparse matrices.
	% and predicate-based index reduces self-connection to the utmost extent~\cite{}.
	\item {\em High parallelizability}:  SM-based storage supports  a complete parallel method since the multiplication of sparse matrices can be computed in a completely parallel way.
\end{itemize}

\section{Query Processing}\label{sec:plan}
In this section, we first introduce that sparse matrix multiplication and join operations can be transformed into each other. Then the SPARQL query graph is converted into a series of sparse matrix multiplications. Again, the basis for the converting of the query graph is given. Finally, the specific implementation of the sparse matrix-based join algorithm is described in detail.

Let $R$ be a relation name and $\mathrm{Sch}(R)$ be the schema of $R$.
If $\mathrm{Sch}(R) = (?x_1,\ldots, ?x_n)$ then we simply denote $R(?x_1,\ldots,?x_n)$ as a $n$-ary relation with the schema of $(?x_1,\ldots,?x_n)$.
The procedure of SM-based BGP query evaluation is generally summarized in the following steps: given a BGP query $Q$,
\begin{enumerate}
	\item Initializing: for every triple pattern $t$ (without loss of generality, we assume that $t$ is of the form $(?x, q, ?y)$) in $Q$, we introduce a binary relation $R_q(?x, ?y)$ and then initialize $R_q(?x, ?y)$ as a $h_p(q)$-index sparse matrix;
	\item Joining: for two relations $R_1(?x_1,\ldots,?x_n)$ and $R_2(?y_1,\ldots,?y_m)$, we introduce a $k$-ary relation $R(?z_1,\ldots,?z_k)$ with $(?x_1,\ldots,?x_n)$ $\cup (?y_1,\ldots,?y_m) = (?z_1,\ldots,?z_k)$ (where $\cup$ is the distinct union) and $R = R_1 \Join R_2$ where $\Join$ is the concatenation (i.e., Join) of relations.
	\item Returning: we output the final relation till all binary relations are concatenated.
\end{enumerate}

\subsection{SM-based Join of Two Triple Patterns}

In this section, we first introduce the multiplication of sparse matrix-matrix.
For the matrix \emph{A}, \emph{a$_{ij}$} represents the item of the \emph{i}th row and the \emph{j}th column of the \emph{A} matrix and \emph{a$_{i*}$} denote the vector consisting of the \emph{i}th row of \emph{A}. Similarly, \emph{a$_{*j}$} represents the vector of the \emph{j}th column of \emph{A}. The operation multiplies a matrix \emph{A} of size $\emph{m}  \times  \emph{v}$ with a matrix \emph{B} of size \emph{v} $\times$ \emph{n} and gives a result matrix \emph{C} of size \emph{m} $\times$ \emph{n}.
In the matrix-matrix multiplication, the \emph{$c_{ij}$} can be defined by
\begin{center}
	\emph{$c_{ij}$}=(\emph{$a_{i*}$} $\cdot$ \emph{$b_{*j}$})
\end{center}
and the \emph{i}th row of the result matrix \emph{C} can be defined by
\begin{center}
	\emph{$c_{i*}$}=(\emph{$a_{i*}$}$\cdot$\emph{$b_{*1}$},$a_{i*}$$\cdot$\emph{$b_{*2}$},$\cdots$,$a_{i*}$$\cdot$\emph{$b_{*j}$},$\cdots$,$a_{i*}$$\cdot$\emph{$b_{*n}$}),
\end{center}
where the operation $\cdot$ is dot product of the two vectors.

\begin{table}[H]
	\begin{center}
		\scalebox{0.85}{
			\begin{floatrow}
				\begin{minipage}{0.2\linewidth}
					\capbtabbox{
						\begin{tabular}{ccc}
							\hline
							row & col &value\\
							\hline
							i & k &1  \\
							i & l &1\\
							\hline
						\end{tabular}
						
					}{
						\caption{A}
						\label{tab:tbA}
					}
				\end{minipage}
				\begin{minipage}{0.2\linewidth}
					\capbtabbox{
						\begin{tabular}{ccc}
							\hline
							row & col&value \\
							\hline
							k & r &1 \\
							k & t &1\\
							l & s &1\\
							l & t &1\\
							\hline
						\end{tabular}
					}{
						\caption{B}
						\label{tab:tbB}
					}
				\end{minipage}
				\begin{minipage}{0.2\linewidth}	
					\capbtabbox{
						\begin{tabular}{ccc}
							\hline
							row & col&value \\
							\hline
							i & r &1 \\
							i & s &1\\
							i & t &2\\
							\hline
						\end{tabular}
					}{
						\caption{C}
						\label{tab:tbC}
					}
				\end{minipage}	
			\end{floatrow}
		}
	\end{center}
\end{table}
Now, we take sparsity of the matrices \emph{A}, \emph{B} and \emph{C} into consideration and describe the multiplication of sparse matrices that only store non-zero items in the form of triples.
We first consider the sparsity of the \emph{A} matrix. Without loss of generality, we assume that the \emph{i}th row of the \emph{A} matrix has only two non-zero entities, respectively in the \emph{k}th and \emph{l}th columns. So, \emph{$a_{i*}$} becomes (\emph{$a_{ik}$},\emph{$a_{il}$}). Since the matrix \emph{B} is sparse as well, again without loss of generality, we assume that the \emph{k}th row of \emph{B} has only two non-zero entries in the \emph{r}th and the \emph{t}th column, and the \emph{l}th row of \emph{B} also has only two non-zero entries in the \emph{s}th and the \emph{t}th column. So the two rows are given by
\emph{$b_{k*}$}=(\emph{$b_{kr}$},\emph{$b_{kt}$}) and \emph{$b_{l*}$}=(\emph{$b_{ls}$},\emph{$b_{lt}$}).
In the process of sparse matrix multiplication, the invalid operation of element 0 is avoided because as long as one of the two elements is 0, the product is also 0. Since the other entities are all zero values, we do not have to explicitly record them, and in the calculation of the vector dot product of the \emph{i}th row of the \emph{C} matrix, the zero entities are directly negligible. So, the \emph{i}th row of the result matrix \emph{C} can also be defined by
\begin{center}
	\emph{$c_{i*}$}=\emph{$a_{ik}$}(\emph{$b_{kr}$},\emph{$b_{kt}$})+\emph{$a_{il}$}(\emph{$b_{ls}$},\emph{$b_{lt}$}).
\end{center}

Now we write \emph{A} and \emph{B} as the triples needed for sparse matrix multiplication.
As shown in Table\ref{tab:tbA}, the triplet form of the sparse matrix of \emph{$a_{i*}$} is described. Similarly, Table \ref{tab:tbB} describes \emph{$b_{k*}$} and \emph{$b_{l*}$}.
In order to find the value of the result matrix \emph{C}, we only need to find the corresponding pairs of elements in \emph{A} and \emph{B} (that is, the $col$ value in \emph{A} and the $row$ value in \emph{B} equal to each other). Thus, in order to obtain a non-zero product, as long as the corresponding non-zero element in \emph{B} is found by multiplying each non-zero element in \emph{A}.

For example, for the \emph{$c_{i*}$} of result matrix \emph{C}, it depends only on \emph{$a_{i*}$} of \emph{A}, and \emph{$a_{i*}$}=(\emph{$a_{ik}$},\emph{$a_{il}$}) is related to \emph{$b_{k*}$} and \emph{$b_{l*}$} of \emph{B} matrix, because the $col$ of \emph{A} corresponds to the $row$ of \emph{B} in the matrix multiplication operation. Traversing all elements of \emph{$a_{i*}$} and finding elements corresponding to \emph{$a_{i*}$} in \emph{B} and multiply their $value$. The first element of \emph{A} is (\emph{i,k},1) and its associated elements are (\emph{k,r},1) and (\emph{k,t},1) in \emph{B}.
To make it easier to understand, let's take the results of these products as (\emph{i,r},1) and (\emph{i,t},1).
The second element of \emph{A} is (\emph{i,l},1) and its associated elements are (\emph{l,s},1) and (\emph{l,t},1) in \emph{B}. its results are (\emph{i,s},1) and (\emph{i,t},1).

Since the value of each element of the \emph{C} matrix is a cumulative value, the result of each multiplication is only a partial value of an element of the \emph{C} matrix. And then the results of these products are cumulatively added according to the $col$ of the \emph{B} matrix. So, two (\emph{i,t},1)  become (\emph{i,t},2) after accumulation.
Because the matrix \emph{C} is also sparse and the \emph{i}th row of \emph{C} only has three nonzero entries in the \emph{r}th, the \emph{s}th and the \emph{t}th column, the row can be given by
\begin{center}
	\emph{$c_{i*}$}=(\emph{$c_{ir}$},\emph{$c_{is}$},\emph{$c_{it}$}),
\end{center}
where \emph{$c_{ir}$}=\emph{$a_{ik}$}\emph{$b_{kr}$}, \emph{$c_{is}$}=\emph{$a_{il}$}\emph{$b_{ls}$} and \emph{$c_{it}$}=\emph{$a_{ik}$}\emph{$b_{kt}$}+\emph{$a_{il}$}\emph{$b_{lt}$}.
The final result matrix \emph{C} is shown in Table \ref{tab:tbC}.

SPARQL join operation is essentially the same as the Matrix multiplication, except that the result representation is not the same.
The matching condition when multiplying two matrices is equivalent to the join variable of the join operation of two tables.
As can be seen from Table \ref{tab:tbC}, the matrix multiplication result representation always has the same format ($row$, $col$, $value$). And each element is the cumulative value of multiple partial product results.

However, the result of the join operation is not a fixed format.
Such as, (\emph{i,k}) join (\emph{k,r}) and (\emph{k,t}) are (\emph{i,k,r}) and (\emph{i,k,t}), (\emph{i,l}) join (\emph{l,s}) and (\emph{l,t}) are (\emph{i,l,s}) and (\emph{i,l,t}), (\emph{i,k,t}) and (\emph{i,l,t}) can't accumulate since they do not represent the same semantic result. Because the two tables join will increase the number of columns (\emph{A}.$row$, \emph{A}.$col$, \emph{B}.$col$) to fully represent this semantics.
The join operation of the two tables is converted to multiplication of two matrices, but only the boolean matching is made, and all of the $value$ are 0, so the stored $value$ column can be omitted.
So, the result matrix \emph{C} of \emph{A} join \emph{B} becomes 3 columns.
When it is necessary to continue the join operation, it is only necessary to regard the first column of the result table as the row stored by the sparse matrix, and then perform the matching according to the join condition.

\subsection{Query Planner}\label{sec:planer}
As discussed above, the join of two triple patterns can be transformed into the multiplication of two sparse matrices.
Therefore, a query graph of multiple triple patterns can also be converted into a series of sparse matrix multiplications.

In general, the performance of query evaluation depends on the order of sparse matrix multiplications since different orders bring intermediate results with a different size.

To generate an optimal query plan, we optimize query evaluation by reordering triple patterns of queries.
As we know, the worst case of join is Cartesian product since the maximum of intermediate results (upper bound) will be brought forth. In the reordering of our optimization, we directly disallow the worst case as far as possible with minimal upper bound. In this case, our optimization is aslo called \emph{connected minimal upper bound}.

To do so, we introduce the two-step procedure: \emph{rearranging triple patterns} and \emph{generating query plan}.

\subsubsection{Rearranging Triple Patterns}
The main strategy of rearranging triple patterns is \emph{connected minimal upper bound} with no changed edges of the original query graph.

The procedure of rearranging triple patterns contains the following steps.
\begin{enumerate}
	\item Sorting all $n$ edges according to label (predicate) statistics and then obtaining a list of edges: $L= (e_1, e_2, \ldots, e_n)$, where $e_i \preceq e_j$ for any $1\le i< j\le n$.
	\item Initializing a collection of nodes of edges, \emph{N} = $\emptyset$, and selecting the smallest edge $e_1$ in the collection \emph{L} and adding its nodes $e_1.node$ into the collection \emph{N}. Then updating the list \emph{L}=($e_2$, $\ldots$ , $e_n$) and the collection \emph{N} = \{$e_1.node$\}, where $e_1.node$ are the nodes of $e_1$.
	\item Selecting an edge $e_i$ in order from the updated list \emph{L} and at least one node of $e_i$ must appear in the updated collection \emph{N}. Then adding its nodes $e_i.node$ into the collection \emph{N} and removing selected $e_i$ from the list \emph{L}. This step guarantees the connection: the selected edge and the previous edges have common join variables.
	\item Repeating step 3 until the collection \emph{L} = $\emptyset$.
\end{enumerate}

Taking query $Q_g$ as an example, its optimization steps are shown in the right figure of Figure~\ref{fig:querygraph} .

Based on the above analysis, we can analyze the overall cost by accumulating all intermediate results during the whole of relations joining to evaluate the cost of the query. We analyze the overall cost from the following two aspects.

\underline{1). Order of Join}

Let $R_1$ and $R_2$ be two relations. We use $\delta(R_1, R_2)$ to denote the number of intermediate results of $R_1 \Join R_2$. Generally, we use $\delta(R_1, \ldots, R_n)$ to denote the number of all intermediate results of $R_1\ \Join \ldots  \Join R_n$. Let $|R|$ denote the number of tuples occurring in $R$.
Clearly, we conclude the following.
\begin{prop}\label{prop:2-boundary}
	Let $R_1$ and $R_2$ be two relations. We have $\min(|R_1|, |R_2|) \le \delta(R_1, R_2) \le |R_1| |R_2|$.
\end{prop}
\begin{proof}
	Let $\mathrm{Sch}(R_1) = (?x_1,\ldots,?x_n)$ and $\mathrm{Sch}(R_2) = (?y_1,\ldots,?y_m)$.
	
	Firstly, we discuss the two boundaries as follows:
	
	If $\{?x_1, \ldots,?x_n\} \cap \{?y_1, \ldots,?y_m\} = \emptyset$ then $R_1 \Join R_2 = R_1 \times R_2$ where $\times$ is the Cartesian product of $R_1$ and $R_2$. In this case, $\delta(R_1, R_2) = |R_1| |R_2|$.
	
	If $\{?x_1, \ldots,?x_n\} = \{?y_1, \ldots,?y_m\}$ then $R_1 \Join R_2 = R_1 \cap R_2$.  In this case, $\delta(R_1, R_2) = \min(|R_1|, |R_2|)$.
	
	Otherwise, we have $\min(|R_1|, |R_2|) < \delta(R_1, R_2) < |R_1| |R_2|$ clearly.
\end{proof}
We can generalize Proposition~\ref{prop:2-boundary} for arbitrary multiple relations.
\begin{thm}\label{thm:n-boundary}
	Let $R_1, \ldots, R_n$ be $n$ relations and $(R_1, \ldots, R_n)$ be a list.
	\[
	\min(|R_1|, |R_2|) + \cdots + \min( \min(\ldots\min(|R_1|,|R_2|)\ldots),|R_n|) \le \delta(R_1, \ldots, R_n)
	\]
	\[
	\le |R_1||R_2| + \cdots+|R_1| \cdots |R_n|.
	\]
\end{thm}
By Theorem~\ref{thm:n-boundary}, we can find that $|R_1||R_2| + \cdots+|R_1| \cdots |R_n|$ is a upper bound of all joining intermediate results of $\delta(R_1, \ldots, R_n)$.

Next, we give a minimal upper bound of $\delta(R_1, \ldots, R_n)$ in the following:
\begin{thm}\label{thm:min-boundary}
	Let $R_1, \ldots, R_n$ be $n$ relations. If $|R_1| \le \cdots \le |R_n|$ then for any list $(R_{i_1}, \ldots R_{i_n})$ of $\{R_1, \ldots, R_n\}$, we have
	\[
	\delta(R_1, \ldots, R_n) \le |R_{i_1}||R_{i_2}| + \cdots+|R_{i_1}| \cdots |R_{i_n}|.
	\]
\end{thm}

Note that the minimal upper bound defined in Theorem~\ref{thm:min-boundary} concerns mainly the scale of relations.

\underline{2). Avoiding the Cartesian Product}

However, in the actual query process, the Cartesian product cost is very large.
If the Cartesian product is not avoided during the query, the largest intermediate result will be generated, which will produce a lot of meaningless results. And it will increase unnecessary I/O time of the query.
Therefore, we should avoid the cost of Cartesian product.

\subsubsection{Generating Query Plan}
Rearranging triple patterns to get an optimized query. The order in which the queries are rearranged is the order in which the joins are executed, that is, the specific physical query plan.

The query optimization algorithm (Query Plan Generated) is listed in Algorithm \ref{alg:QueryPlan}.
In the first line, sorting all edges list \emph{L} in the query graph according to label (predicate) statistics. In lines 2-4, selecting an edge with the smallest number of statistics $e_1$ as the starting edge, setting its nodes as the initial collection \emph{N} and removing selected edge from list \emph{L}. In lines 5-12, selecting an edge from the updated list \emph{L} in order, which has at least one node that are elements of \emph{N}. If the list \emph{L} is not empty, it means that all edges are not covered, so this step is performed cyclically.

\begin{algorithm}\label{alg:QueryPlan}
	\caption{Query Plan Generated}%
	\LinesNumbered %
	\KwIn{List: \emph{L} = ($e_1$, $e_2$, $\ldots$ , $e_n$), Collection: \emph{N} = $\emptyset$ }%
	\KwOut{List: Optimized Query Plan \emph{Q}}%
	%tp.map(); \qquad // HashMap:Predicates Statistics\\
	\emph{L.sort} ( ) \; %\qquad // Sort based on predicates statistics\\
	\emph{Q.add} ( $e_1$ ) \;
	\emph{N} = \emph{N} $\cup$ \{ $e_1.node$\}\;
	\emph{L.remove} ( $e_1$ )\;
	
	\While{ L.size > 0}{
		i=1\;
		\eIf{$e_i.node$ $\in$ \emph{N}}{
			\emph{N}=\emph{N} $\cup$ \{ $e_i.node$ \}\;
			Q.add ( $e_i$ )\;
		}{
			i++;
		}
		
		\emph{L}.remove ( $e_i$ )\;
	}
\end{algorithm}
\subsection{Query Execution}\label{sec:exe}

The SPARQL query is transformed into the multiplications of a series of sparse matrices. Therefore, table-based join operations are converted to sparse matrix-based join (SM-based join) that approximates multiplications for sparse matrices.
After obtaining an optimized query, a series of sparse matrices can be obtained based on the predicates of each triple pattern in the query (assuming the predicate is not a variable, which is known).

Next, we illustrate the SM-based join operation through an example.
As for the BGP as Figure \ref{fig:query}, the triple patterns of SPARQL query are optimized and transformed into an optimal query shown in Figure \ref{fig:BGP}.
Then based on the optimized query, the SM-based join operation is performed.

\begin{figure}[H]
	\centering
	\begin{center}
		\fbox{
			\parbox{2.5cm}{
				\hspace*{0.15in} ?X likes   ?W \\		
				\hspace*{0.15in} ?Z likes   ?W \\
				\hspace*{0.15in} ?X follows ?Y \\
				\hspace*{0.15in} ?Y follows ?Z 
			}
		}
	\end{center}
	\caption{BGP Sequence}	\label{fig:BGP}
\end{figure}
\begin{figure}[H]
	\centering
	\includegraphics[width=1\textwidth]{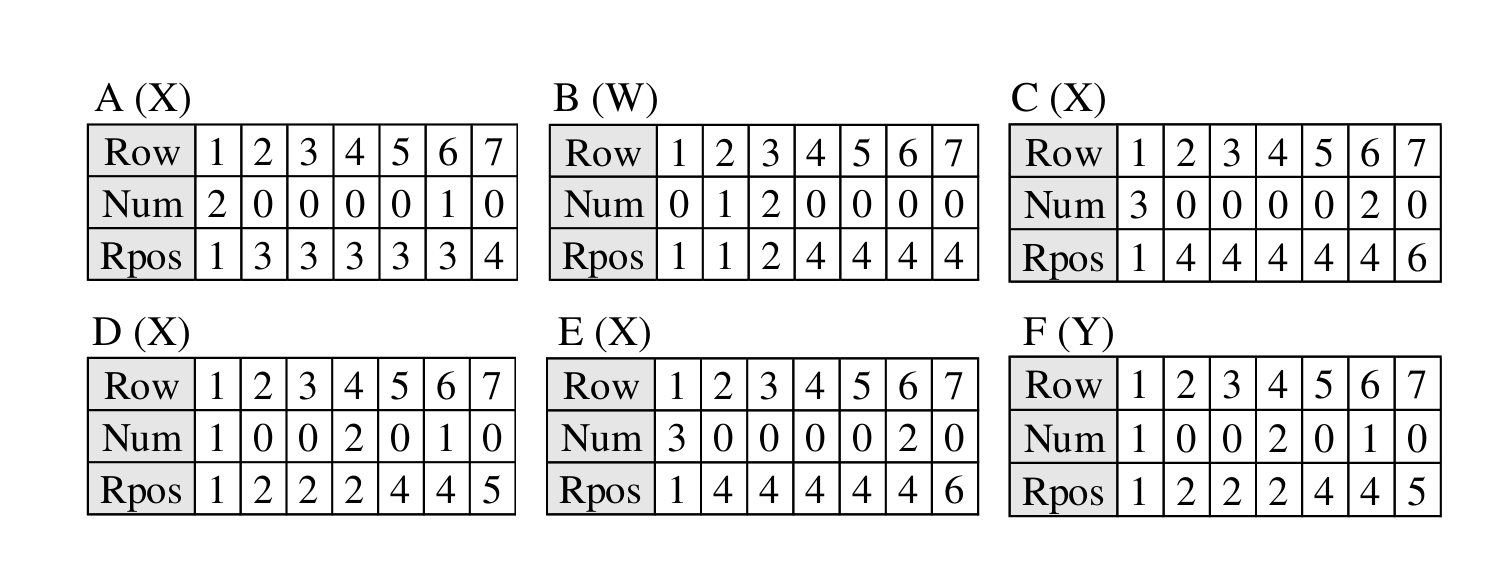}
	\caption{Auxiliary arrays}\label{fig:multiply}
\end{figure}

Figure~\ref{fig:queryprocess} shows the SM-based join processing. Tables $A$, $B$, $D$ and $F$ represent four sparse matrices obtained from the query triple patterns in Figure~\ref{fig:BGP} and the RDF dataset in Figure~\ref{fig:rdf}. And table $C$, $E$ and $G$ represent the intermediate results during the join processing. The first row of each table represents variable bindings of triple pattern and every join is a process of sparse matrix multiplication.

To assist SM-based join, we create an auxiliary array for every sparse matrix.
This is a necessary step for the join operations of these two sparse matrices, $SM_L$ and $SM_{R}$. 
The auxiliary arrays, \emph{$SM_L$}($\alpha$) and $SM_{R}$($\beta$), can quickly index to the start position of each non-zero row, and can count the number of non-zero entries in each row.
Creating auxiliary arrays of $SM_{L}$ and $SM_{R}$:
For $SM_{L}$, we build the auxiliary array according to the row $\alpha$: the non-zero items and the starting offsets of each row are counted.
For $SM_{R}$, we build the auxiliary array according to the row $\beta$ of the first join variable with $SM_{L}$: the non-zero items and the starting offsets of each row are counted. The row of \texttt{Row} represents the row number derived from the sparse matrix, \texttt{Num} represents the number of all non-zero entries per \texttt{Row}, \texttt{Rpos} represents the starting position of the first non-zero entry per \texttt{Row} in the sparse matrix.
As shown in Figure~\ref{fig:multiply}, the auxiliary arrays A(X) and  B(W) of sparse matrices $A$ and $B$ shown in Figure~\ref{fig:queryprocess} are established.

\begin{figure}[H]
	\centering
	\includegraphics[width=1\textwidth]{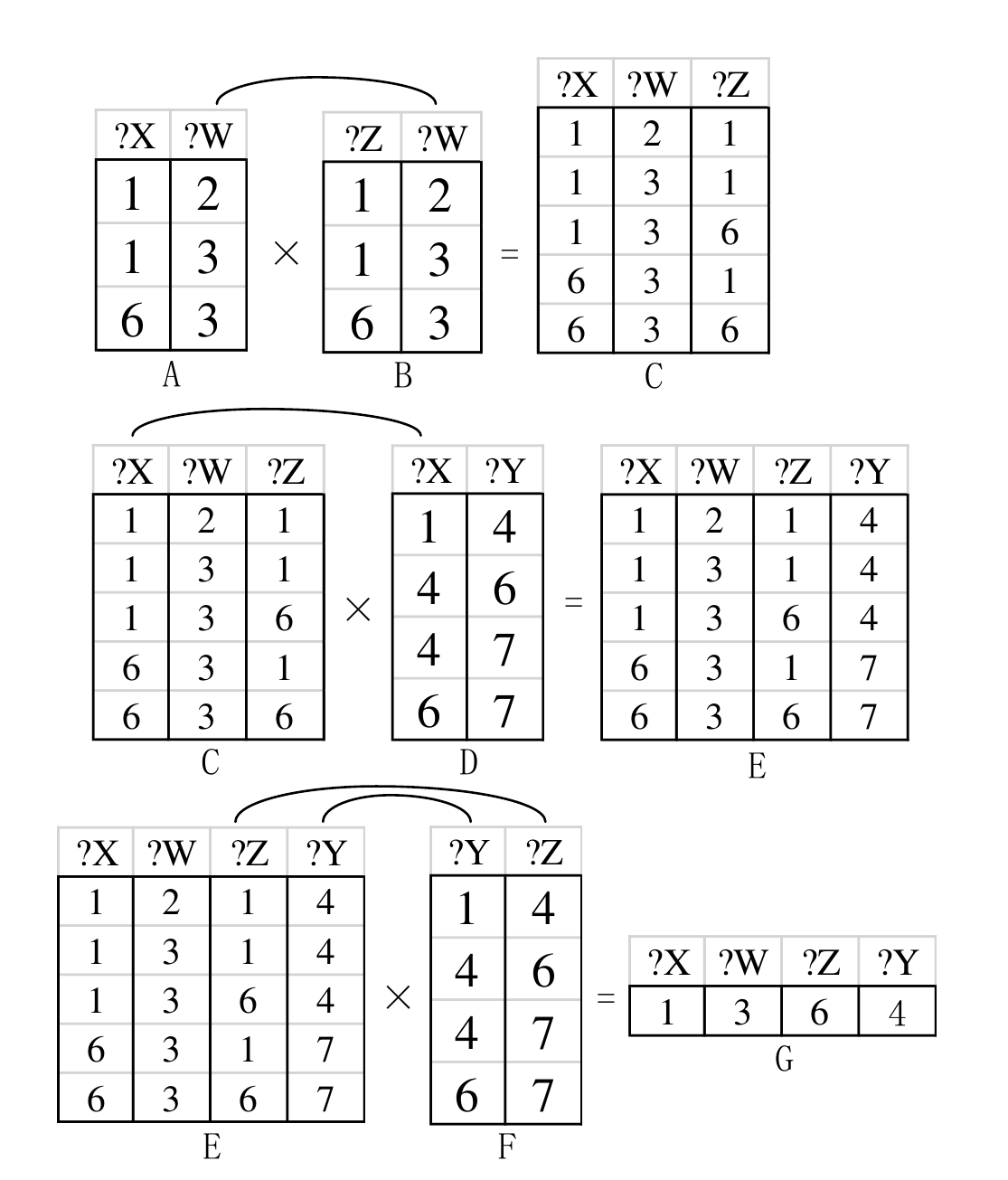}
	\caption{The representation of the query processing}\label{fig:queryprocess}
\end{figure}

We traverse the auxiliary arrays A(X) and B(W) to collect all the elements in the row \texttt{Num} which is large than 0.
Obviously, the columns in A(X) we get are \{1,2,1\} and \{6,1,3\}, which means there exists 2 elements in \texttt{Row} 1 starting from the first position and 1 element in \texttt{Row} 6 starting from the third position in the sparse matrix $A$. Thus, we traverse the first position until the second position to obtain the element that belonged to row 1, i.e., \{1,2\} and \{1,3\} and the third position to obtain the element that belonged to row 6, i.e., \{6,3\} in $A$. Further, based on the value of ?W, the common variable in $A$ and $B$, we index the column 'value' in the row \texttt{Row} in auxiliary array B(W) whose \texttt{NUM} is larger than 0, i.e., existing non-zero entries. For example, based on the 'value' 2 of $column$ ?W of the entry \{1,2\} in $A$, we find out the entry \{2,1,1\} of B(W), which is matched against the first position of $B$, i.e., \{1,2\} to obtain the result \{1,2,1\} in $C$.

The query execution algorithm is listed in Algorithm \ref{alg:QueryExecution}. The algorithm describes the BGP query of SPARQL.
A SPARQL query is a continuous SM-based join operation. For a whole SPARQL query, we only need to loop join operation in a query plan.
In the first line, the first sparse matrix is obtained based on the optimized query plan $Q$.
In lines 2-4, if there are suspended triple patterns in the query plan, it first obtains its sparse matrix and joins the previous intermediate results.
All triple patterns in the query plan $Q$ are executed in sequence until all tasks are executed.

\begin{algorithm}[t]\label{alg:QueryExecution}
	\caption{Query Execution}%算法名字
	\LinesNumbered %要求显示行号
	\KwIn{Query: \emph{Q}}%输入参数
	\KwOut{Result table: \emph{Rt}}%输出
	\emph{Rt} $\longleftarrow$ the first sparse matrix of \emph{Q} \;
	%Reading sparse table \emph{Rt} according to the first \emph{tp} of \emph{Q}\;
	\While{\emph{Q}.empty != True}{
		\emph{SpTable} $\longleftarrow$ the next sparse matrix of \emph{Q} \;
		%Reading sparse table \emph{SpTable} according to the next \emph{tp} of \emph{Q}\;
		\emph{Rt} = \emph{Join( Rt, SpTable )}\;
	}
\end{algorithm}

The specific SM-based join operations are described in Algorithm~\ref{alg:SMjoin}.
In the first line, we first create respective auxiliary arrays based on two sparse matrices.
In line 4, all non-zero rows of sparse matrix $A$ are traversed according to the auxiliary array of $A$.
In line 5, traversing all non-zero entries of each non-zero row in the sparse matrix $A$ according to the auxiliary array of $A$.
In line 6, for each non-zero item in $A$, finding all non-zero items in $B$ that has at least one common join variable with non-zero item of $A$ through the auxiliary array of $B$.
In lines 7-10, if other join variables are also matched, the result will be outputed. Otherwise, matching the next one.
And the result is (1,3,6,4), it is anti-hashed to the original string (A,I2,C,B). The results of the original dataset are shown in the red part of the Figure~\ref{fig:rdf}.

\begin{algorithm}[t]\label{alg:SMjoin}
	\caption{SM-based Join( \emph{A}, \emph{B} )}
		\LinesNumbered 
	\KwIn{Two sparse table: \emph{A}, \emph{B}}
	\KwOut{Result table: \emph{Rt}}
	\emph{Building auxiliary arrays of A and B} \;
	\emph{NR} $\longleftarrow$ all non-zero rows \\
	\emph{NE} $\longleftarrow$ all non-zero entries of non-zero row\\
	\For{ Nr\_A in NR\_A }{
		
		\For{ Ne\_A in Nr\_A }{
			
			\For{ Ne\_B matched to Ne\_A}{
				
				\eIf{ other join variables are also satisfied}{
					output\;
				}{
					continue\;
				}	
			}
		}	
	}
\end{algorithm}

\section{Extension: Join on GPU}\label{sec:GPU}
In this section, we mainly demonstrate the scalability of gSMat on the GPU. And we just extend the sparse matrix-based join operation to the GPU for implementation with CUDA. It's called gSMat$^+$.

\subsection{GPU Background}

Graphics Processing Units (GPUs) are specialized architectures traditionally designed for game applications.
Recent researches show that they can significantly speed up database query processing.
They are used as co-processors for the CPU.
GPU's performance advantage comes from its high parallelism and high memory bandwidth.
The ``parallel computing'' refers to ``multiple data parallel computing time and a data execution time is the same''. As a result, GPU is more suitable for parallel processing of intensive data.

A general flow for a computation task on the GPU consists of three steps. First, the host code allocates GPU memory for input data and output data, and transfers input data from the host memory to the GPU device memory. Second, the host code starts the kernel on the GPU. The kernel performs the task on the GPU. Third, when the kernel execution is done, the host code transfers results from the GPU device memory to the host memory.
We implemented our SM-based join algorithms using CUDA, Our GPU-based joins can be easily mapped to the CUDA framework.

\subsection{Query Processing on GPU}
Matrix operations are well-suited for computation on the GPU, and matrix multiplication is computationally intensive on the GPU. The parallelization of matrix operations is also suitable for sparse matrix operations, and sparse matrices also save a lot of thread resources because it does not store a large number of extra zero entry.
In this subsection, we first introduce our expansion work based on SM-based join. Then, we detailedly discuss the main step of 
\emph{preallocating result matrix space}.\nop{, because the number of the nonzero entries in the result matrix is unknown in advance.}

\subsubsection{SM-based Join on GPU}
The extension of the SM-based join algorithm on the GPU is actually that it can be implemented on the GPU device in parallel.
Our implementation steps include the following steps:
\begin{enumerate}
	\item \emph{Transferring the required data to GPU.} In this step, the data involved in a given SPARQL query is transferred to the GPU device memory. Marking data that requires multiple multiplication operations avoids duplicate transferring.
	\item \emph{Preallocating result matrix space.} In this step, we can get the number of nonzero entries estimated when two sparse matrices join   and the starting offset to store the non-zero items of each row via Algorithm 4. Then we allocate the memory space required by the result matrix.
	\item \emph{Starting the kernel to perform SM-based join on the GPU.} In this step, each non-zero row is processed by one thread. And the result of each thread calculation is stored in the result matrix according to the starting offset of each row.
	\item \emph{Transferring the result data to host.} In this step, after the join operation is completed, the result matrix is transferred back to the host.
\end{enumerate}

The above is the overall process of extending the SM-based join algorithm on the GPU. It is worth noting that the thread is not allocated to all rows, but only to nonzero rows. In this case, a large number of zero rows will be directly removed during the calculation, saving thread resources.

\subsubsection{Preallocating Result Matrix Space}

In the matrix-matrix multiplication, multiplication of matrices pre-allocates a predictable size matrix and store entries to predictable memory addresses. However, the result matrix of multiplication of sparse matrices simply stores non-zero entries. Because the number of the nonzero entries in the result sparse matrix is unknown in advance, precise memory allocation of the sparse matrix multiplication is impossible before real computation. And physical address of each new result entry is also unknown.
Therefore, in order to estimate the space required by the result matrix as accurately as possible, we pre-calculate the size of the possible non-zero entries of the result matrix before the main code allocates GPU memory. This step is also called \emph{minimize the maximum space required}.

\emph{Minimize the maximum space required}, that is, the number of non-zero items can be generated by the result matrix as accurately as possible. We count the number of nonzero entries per row, and record the starting offset of nonzero entries in each row in the assigned result matrix so that nonzero entries can be written to the result matrix in parallel.

Algorithm~\ref{alg:preallocation} describes this procedure, which calculates the upper bound of the nonzero entries in each row of the result matrix and records the starting offset of non-zero entries of each row in the assigned result matrix.
We create two arrays $N$ and $P$ of size m to record the upper bound sizes of the rows and the starting offset, where m is the number of rows of result matrix. We first use the GPU to calculate each item of the array $N$ in parallel and then calculate the starting offset of the final result per row based on the number of rows.
Experimental results show that this method saves a lot of global memory space.
\begin{algorithm}[t]\label{alg:preallocation}
	\caption{Pre-allocating result matrix space}%算法名字
	\LinesNumbered %要求显示行号
	\KwIn{sparse matrix \emph{A, B}}%输入参数
	\KwOut{array $N$, $P$}%输出
	
	\For{ $N_i$ in $N$ in parallel}{
		$N_i$ = 0 \;
		\For{ each nonzero entry $a_{ij}$ in $a_{i*}$}{
			nnz($b_{j*}$) $\longleftarrow$ number of nonzero $b_{j*}$ joined with $a_{ij}$\;
			$N_i$ = $N_i$ + nnz($b_{j*}$)\;
		}	
	}
	\For{$P_i$ in $P$}{
		$P_i$ = 0 \;
		$P_i$ = $N_{i-1}$ +$P_{i-1}$ \;
	}
\end{algorithm}

\subsection{Optimization on GPU}
Extending on the GPU, there are many factors that affect the efficiency of the implementation. For example, efficient design of software-side parallel algorithms and resource utilization in hardware. We mainly discuss the impact of the following aspects on efficiency.

\emph{\textbf{Data transfer analysis}}

Through experiments we found that the amount of data transferring is the main source of overhead for GPU-based algorithms.
A SPARQL query consists of multiple SM-based join operations, and the same sparse matrix may also be joined multiple times.
When there are multiple identical predicate of triple patterns in a SPARQL query, indicating that the data needs to be transmitted multiple times, the cost of repeatedly transferring data multiple times is also very expensive.
So, our optimization method is to mark data that requires multiple transfer and then transmit the same data only once.

\emph{\textbf{Shared memory and global memory}}

The device memory of the GPU is also limited. When the processed data is very large, the data cannot be transferred to the device memory all at once. This requires single-step transmission or block transmission during the calculation process.
We try to allocate the global memory as large as possible to store the needed data, such that the number of transferring is reduced. At the same time, shared memory can be used for small results since the shared memory size is also limited.
Although shared memory is a few orders of magnitude smaller than global memory, shared memory accesses faster than global memory.

\emph{\textbf{Other costs}}

In addition to the above optimizations, there are unavoidable costs as follows.\\
\emph{Preheating of GPU:} there is large startup time (preheating overhead) for the first time when the GPU is called by kernel function.\\
\emph{Kernel startup cost:} each Kernel function starts with time, that is, Kernel startup overhead.

\section{Evaluation}\label{sec:ex}
In this section we evaluate gSMat and gSMat$^+$ against some existing popular RDF stores using the known RDF benchmark datasets. We choose RDF-3X, gStore for evaluation, since they show much better performance than other RDF stores.

\subsection{Experimental Setup}
The experiments are conducted on three RDF benchmarks. The synthetic data sets come from well-established Semantic Web benchmarks: the Waterloo SPARQL Diversity Test Suite (WatDiv) and the real world data sets correspond to open source YAGO and DBpedia.
We use five different sizes of synthetic data , which are 100 million, 200 million, 300 million, 400 million and 500 million RDF triples generated by the WatDiv Data Generator with scale factors 1000, 2000, 3000, 4000 and 5000, respectively. It is provided by WatDiv that covers all different query shapes and thus allows us to test the performance of gSMat and gSMat$^+$.
YAGO is a real RDF dataset which consists of facts extracted from Wikipedia (exploiting the infoboxes and category system of Wikipedia) and integrated with the WordNet thesaurus. The YAGO dataset contains 200,737,655 distinct triples and 38,734,252 distinct strings, consuming 14 GB as (factorized) triple dump.
DBPedia is another real RDF set, which constitutes the most important knowledge base for the Semantic Web community. Most triples in DBpedia comes from the Wikipedia Infobox.
The data characteristics are summarized in Table \ref{tab:dataattributes}.
\begin{table}[!htbp]
	
	\caption{Benchmark Statistics} \label{tab:dataattributes}
	\begin{tabular}{|c |c |c |c |}
		\hline
		\#Dataset	& \#Triples & \#(S $\cap$ O)	& \#P \\
		\hline
		WatDiv100M	& 108 997 714	& 10 250 947	& 86 \\
		\hline
		WatDiv200M	& 219 783 842	& 20 296 483	& 86 \\
		\hline
		WatDiv300M	& 329 827 477	& 30 221 812	& 86 \\
		\hline
		WatDiv400M	& 439 433 765	& 40 040 420	& 86 \\
		\hline
		WatDiv500M	& 549 246 141	& 49 771 433	& 86 \\
		\hline
		Yago	& 200 737 655	& 38 734 252	& 46 \\
		\hline
		DBpedia	& 120 978 080	& 42 966 066	& 4 282 \\
		\hline
		
	\end{tabular}
\end{table}

gSMat is implemented in C and compiled by using GCC. The SM-based join algorithm on GPU is implemented in CUDA C and compiled by using NVCC. The experimental environment is shown in Table \ref{tab:Experimentalenvironment}.

\begin{table}[!htbp]
	\caption{Experimental Environment}\label{tab:Experimentalenvironment}
	{\small 	\begin{tabular}{l |l }
			\hline
			Configuration & Machine \\
			\hline
			CPU & Intel(R) Xeon(R) E5-2603 v4 @ 1.70GHz \\
			\hline
			System memory & 72GB \\
			\hline
			GPU & NVIDIA Tesla M40(3072 CUDA Cores,1.11GHz) \\
			\hline
			GPU memory & 24GB \\
			\hline
			Syatem software & Ubuntu 14.04.5 LTS, GPU driver version 367.48,\\
			and library & CUDA Driver Version / Runtime Version 8.0/8.0 \\
			\hline
			
	\end{tabular}}
\end{table}

\subsection{Experiments on Synthetic Datasets}
In this subsection, using WatDiv, we test the performances of our method in two aspects: the efficiency and the scalability.
Here, we generate the query workloads from the respective RDF datasets, which are available as RDF triplesets.
The WatDiv benchmark defines 20 query templates classified into four categories: linear (L), star (S), snowflake (F) and complex queries (C).
In order to avoid the effect of caches on experimental results, we drop the caches before each execution.
Query times are averaged over 10 consecutive runs. All results are rounded to 1 decimal place.
gStore cannot handle datasets of more than 300 million triples in our environment, so WatDiv400M and WatDiv500M over gStore cannot be included in the results of the experiments.

\subsubsection{Efficiency Test}

First, we analyze the query efficiency according to the average query time of the four types. Due to page limit, we only show partial results (WatDiv100M, WatDiv300M and WatDiv500M) on different query engines as shown in Tables~\ref{tab:wat100}, \ref{tab:wat300} and \ref{tab:wat500} respectively.

\begin{table}[!htbp]
	\scriptsize
	\centering
	\caption{Runtime for WatDiv100M queries (ms)}
	\label{tab:wat100}
	\resizebox{\textwidth}{!}{%
		\begin{tabular}{|l|l|l|l|l|}
			\hline
			Wat100 & C       & F     & L   & S   \\
			\hline
			gSMat$^+$ & 3881.3 & 2950.6 & 1313.2  & 1812.6 \\
			\hline
			gSMat & 6678.1 & 3784.1 & 1932.5  & 2503.3 \\
			\hline	
			RDF-3X   & 11980.3 & 8405.6 & 16282.1  & 3820.6 \\
			\hline
			gStore   & 15447.0  & 29204.8 & 20128.0 & 10808.7 \\
			\hline
		\end{tabular}
	}
\end{table}

\begin{table}[!htbp]
	\scriptsize
	\centering
	\caption{Speedup for WatDiv100M queries}
	\label{tab:speedupwat100}
	\resizebox{\textwidth}{!}{%
		\begin{tabular}{|l|l|l|l|l|}	
			\hline
			Wat100 & C       & F     & L   & S   \\	
			\hline
			gSMat/RDF-3X   & 1.79 & 2.22 & 8.43  & 1.53 \\
			\hline
			gSMat$^+$/RDF-3X   & 3.09 & 2.85 & 12.40  & 2.11 \\
			\hline
			gSMat/gStore   & 2.31 & 7.72 & 10.42  & 4.32 \\
			\hline
			gSMat$^+$/gStore   & 3.98 & 9.90 & 15.33  & 5.96 \\
			\hline
		\end{tabular}
	}
\end{table}

\begin{table}[!htbp]
	\centering
	\caption{Runtime for WatDiv300M queries (ms)}
	\label{tab:wat300}
	\scriptsize
	\resizebox{\textwidth}{!}{%
		\begin{tabular}{|l|l|l|l|l|}
			\hline
			Wat300 & C       & F     & L   & S   \\
			\hline
			gSMat$^+$ & 10597.6 & 8136.3 & 3657.7  & 4794.2 \\
			\hline
			gSMat & 18106.2 & 9185.3 & 5033.3  & 5350.3 \\
			\hline
			RDF-3X   & 39571.1 & 26231.3 & 56879.2  & 16066.5 \\
			\hline
			gStore   & 111490.8  & 1047971.5 & 311390.4 & 296687.5 \\
			\hline
		\end{tabular}
	}
\end{table}		
\begin{table}[!htbp]
	\centering
	\caption{Speedup for WatDiv300M queries}
	\label{tab:speedupwat300}
	\scriptsize
	\resizebox{\textwidth}{!}{%
		\begin{tabular}{|l|l|l|l|l|}
			\hline
			Wat300 & C       & F     & L   & S   \\		
			\hline
			gSMat/RDF-3X   & 2.19 & 2.86 & 11.30  & 3.00 \\
			\hline
			gSMat$^+$/RDF-3X   & 3.73 & 3.22 & 15.55  & 3.35 \\
			\hline
			gSMat/gStore   & 6.16 & 114.09 & 61.87  & 55.45 \\
			\hline
			gSMat$^+$/gStore   & 10.52 & 128.80 & 85.13  & 61.88 \\
			\hline
		\end{tabular}
	}
\end{table}

\begin{table}[!htbp]
	\centering
	\caption{Runtime for WatDiv500M queries (ms)}
	\label{tab:wat500}
	\scriptsize
	\resizebox{\textwidth}{!}{%
		\begin{tabular}{|l|l|l|l|l|}
			\hline
			Wat500 & C       & F     & L   & S   \\
			\hline
			gSMat$^+$ & 17375.5 & 13958.5 & 5746.0  & 7769.5 \\
			\hline
			gSMat & 35509.1 & 16071.4 & 9719.6  & 9072.4 \\	
			\hline
			RDF-3X   & 66513.7 & 46906.9 & 92685.6  & 27456.7 \\
			\hline
		\end{tabular}
	}
\end{table}

\begin{table}[!htbp]
	\centering
	\caption{Speedup for WatDiv500M queries }
	\label{tab:speedupwat500}
	\scriptsize
	\resizebox{\textwidth}{!}{%
		\begin{tabular}{|l|l|l|l|l|}
			\hline
			Wat500 & C       & F     & L   & S   \\		
			\hline
			gSMat/RDF-3X   & 1.87 & 2.92 & 9.54  & 3.03 \\
			\hline
			gSMat$^+$/RDF-3X   & 3.83 & 3.36 & 16.13  & 3.53 \\
			\hline
		\end{tabular}
	}
\end{table}

Table~\ref{tab:speedupwat100}, \ref{tab:speedupwat300} and \ref{tab:speedupwat500} show speedup of RDF-3X and gStore over WatDiv datasets where
gSMat/RDF-3X and gSMat$^+$/RDF-3X represent the speedup of gSMat and gSMat$^+$ with respect to RDF-3X while gSMat/gStore and gSMat$^+$/gStore represent the speedup of gSMat and gSMat$^+$ with respect to gStore, respectively.

The first observation is that gSMat performs much better than RDF-3X and gStore for all queries. And the performance of gSMat$^+$ is better than gSMat after parallel extension on GPU.
We first analyze and discuss the comparison with RDF-3X. The most obvious acceleration effect is the L type queries, up to 11 times for gSMat and up to 16.13 times for gSMat$^+$.
For complex type queries, the speedup is generally stable at about 2 times for gSMat and 3.7 times for gSMat$^+$.
For other type queries (star and snowFlake), the speedup is generally stable at about 3 times for gSMat and gSMat$^+$.
Next we discuss the comparison with gStore. Due to our environmental constraints, we can only analyze gStore with triple size up to 300 million. From the experimental result, we can see that gSMat$^+$ and gSMat have a great improvement for different type queries of data in gStore.
In our environment, when the data is a maximum of 300 million, the acceleration ratios of snowflake, linear, star, and complex types are 114.09, 61.87, 55.45, 6.16 for gSMat and 128.80, 85.13, 61.88, 10.52 for gSMat$^+$, respectively.

%Average query runtimes for WatDiv datasets
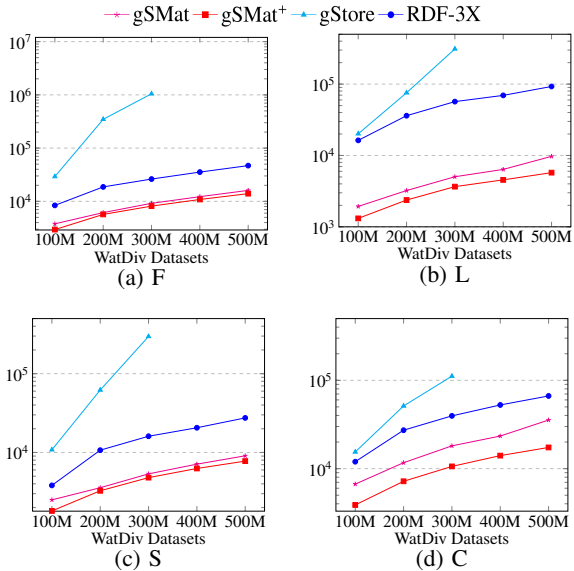
\begin{figure}[H]
	%**********************F
	\includegraphics[width=0.6\textwidth]{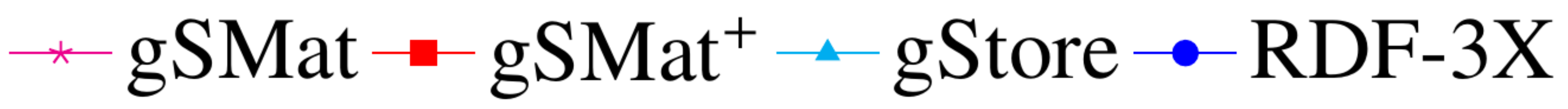}
	\subfigure[F]{\label{fig:wat-F}
		\begin{minipage}[htbp]{0.45\linewidth}
			\centering
			\scalebox{0.45}{
				\begin{tikzpicture}[font=\Large]
				\begin{semilogyaxis}[
				symbolic x coords={100M,200M,300M,400M,500M},
				bar width=10pt,
				xlabel={WatDiv Datasets},
				ymin=0,
				legend style={at={(0.35,0.5)}},
				ymajorgrids=true,
				grid style=dashed,
				anchor=north,
				legend pos = north west,
				ymin=2900,ymax=12000000,
				xlabel style={below=-0.02cm},
				%  legend style={at={(-0.132, -0.2)}}, %the coordinate of legend
				legend columns=1, %%%%make the legend at the below.
				ylabel near ticks,
				]
				
				%gSMat_CPU
				\addplot[
				color=magenta,
				mark=star, mark options={fill=magenta}
				]
				coordinates {
					(100M,3784.1) (200M,6145.5) (300M,9185.3) (400M,12240.3) (500M,16071.4)
				};
				
				%gSMat_GPU
				\addplot[
				color=red,
				mark=square*, mark options={fill=red}
				]
				coordinates {
					(100M,2950.6) (200M,5684.9) (300M,8136.3) (400M,10850.0) (500M,13958.5)
				};
				
				%gStore
				\addplot[
				color=cyan,
				mark=triangle*, mark options={fill=cyan}
				]
				coordinates {
					(100M,29204.8) (200M,347424.1) (300M,1047971.5)
				};
				
				%   RDF3X
				\addplot[
				color=blue,
				mark=*, mark options={fill=blue}
				]
				coordinates {
					(100M,8405.6) (200M,18681.2) (300M,26231.3) (400M,35474.7) (500M,46906.9)
				};
				%\legend{gSMat,gSMat$^+$,gStore,RDF-3X}
				\end{semilogyaxis}
				\end{tikzpicture}
			}
			\label{fig:queryWatdiv-F}
		\end{minipage}
	}
	\subfigure[L]{\label{fig:wat-L}
		\begin{minipage}[htbp]{0.45\linewidth}
			\centering
			\scalebox{0.45}{
				\begin{tikzpicture}[font=\Large]
				\begin{semilogyaxis}[
				symbolic x coords={100M,200M,300M,400M,500M},
				bar width=10pt,
				xlabel={WatDiv Datasets},
				ymin=0,
				legend style={at={(0.35,0.5)}},
				ymajorgrids=true,
				grid style=dashed,
				anchor=north,
				legend pos = north west,
				ymin=1000,ymax=500000,
				xlabel style={below=-0.02cm},
				%  legend style={at={(-0.132, -0.2)}}, %the coordinate of legend
				legend columns=2, %%%%make the legend at the below.
				ylabel near ticks,
				]
				
				%gSMat_CPU
				\addplot[
				color=magenta,
				mark=star, mark options={fill=magenta}
				]
				coordinates {
					(100M,1932.5) (200M,3224.2) (300M,5033.3) (400M,6379.2) (500M,9719.6)
				};
				
				%gSMat_GPU
				\addplot[
				color=red,
				mark=square*, mark options={fill=red}
				]
				coordinates {
					(100M,1313.2) (200M,2373.8) (300M,3657.7) (400M,4552.8) (500M,5746.0)
				};
				
				%gStore
				\addplot[
				color=cyan,
				mark=triangle*, mark options={fill=cyan}
				]
				coordinates {
					(100M,20128.0) (200M,75573.2) (300M,311390.4)
				};
				
				%  RDF3X
				\addplot[
				color=blue,
				mark=*, mark options={fill=blue}
				]
				coordinates {
					(100M,16282.1) (200M,36088.7) (300M,56879.2) (400M,69532.5) (500M,92685.6)
				};
				%\legend{gSMat,gSMat$^+$,gStore,RDF-3X}
				\end{semilogyaxis}
				\end{tikzpicture}
			}
			\label{fig:queryWatdiv-L}
		\end{minipage}
	}
	
	\subfigure[S]{\label{fig:wat-S}
		\begin{minipage}[htbp]{0.45\linewidth}
			\centering
			\scalebox{0.45}{
				\begin{tikzpicture}[font=\Large]
				\begin{semilogyaxis}[
				symbolic x coords={100M,200M,300M,400M,500M},
				bar width=10pt,
				xlabel={WatDiv Datasets},
				ymin=0,
				legend style={at={(0.35,0.5)}},
				ymajorgrids=true,
				grid style=dashed,
				anchor=north,
				legend pos = north west,
				ymin=1800,ymax=500000,
				xlabel style={below=-0.02cm},
				%  legend style={at={(-0.132, -0.2)}}, %the coordinate of legend
				legend columns=1, %%%%make the legend at the below.
				ylabel near ticks,
				]
				
				%gSMat_CPU
				\addplot[
				color=magenta,
				mark=star, mark options={fill=magenta}
				]
				coordinates {
					(100M,2503.3) (200M,3568.5) (300M,5350.3) (400M,7108.0) (500M,9072.4)
				};
				
				%gSMat_GPU
				\addplot[
				color=red,
				mark=square*, mark options={fill=red}
				]
				coordinates {
					(100M,1812.6) (200M,3265.7) (300M,4794.2) (400M,6265.6) (500M,7769.5)
				};
				
				%gStore
				\addplot[
				color=cyan,
				mark=triangle*, mark options={fill=cyan}
				]
				coordinates {
					(100M,10808.7) (200M,62014.8) (300M,296687.5)
				};
				
				%  RDF3X
				\addplot[
				color=blue,
				mark=*, mark options={fill=blue}
				]
				coordinates {
					(100M,3820.6) (200M,10693.4) (300M,16066.5) (400M,20593.2) (500M,27456.7)
				};
				%\legend{gSMat,gSMat$^+$,gStore,RDF-3X}
				\end{semilogyaxis}
				\end{tikzpicture}
			}
			\label{fig:queryWatdiv-S}
		\end{minipage}
	}
	\subfigure[C]{\label{fig:wat-C}
		\begin{minipage}[htbp]{0.45\linewidth}
			\centering
			\scalebox{0.45}{
				\begin{tikzpicture}[font=\Large]
				\begin{semilogyaxis}[
				symbolic x coords={100M,200M,300M,400M,500M},
				bar width=10pt,
				xlabel={WatDiv Datasets},
				ymin=0,
				legend style={at={(0.35,0.5)}},
				ymajorgrids=true,
				grid style=dashed,
				anchor=north,
				legend pos = north west,
				ymin=3300,ymax=500000,
				xlabel style={below=-0.02cm},
				%  legend style={at={(-0.132, -0.2)}}, %the coordinate of legend
				legend columns=1, %%%%make the legend at the below.
				ylabel near ticks,
				]
				
				%gSMat_CPU
				\addplot[
				color=magenta,
				mark=star, mark options={fill=magenta}
				]
				coordinates {
					(100M,6678.1) (200M,11671.5) (300M,18106.2) (400M,23384.5) (500M,35509.1)
				};
				
				%gSMat_GPU
				\addplot[
				color=red,
				mark=square*, mark options={fill=red}
				]
				coordinates {
					(100M,3881.3) (200M,7189.8) (300M,10597.6) (400M,14044.4) (500M,17375.5)
				};
				
				%gStore
				\addplot[
				color=cyan,
				mark=triangle*, mark options={fill=cyan}
				]
				coordinates {
					(100M,15447.0) (200M,51204.8) (300M,111490.8)
				};
				
				%   RDF3X
				\addplot[
				color=blue,
				mark=*, mark options={fill=blue}
				]
				coordinates {
					(100M,11980.3) (200M,27237.8) (300M,39571.1) (400M,52600.5) (500M,66513.7)
				};
				%\legend{gSMat,gSMat$^+$,gStore,RDF-3X}
				\end{semilogyaxis}
				\end{tikzpicture}
			}
			\label{fig:queryWatdiv-C}
		\end{minipage}
	}
	\vspace*{-0.6cm}
	\caption{Average query runtime for WatDiv datasets (ms)} \label{fig:queryWatdiv}
\end{figure}

\subsubsection{Scalability Test}

In this experiment, we study how the performances scale with the size of data.
A comparison graph of the average query time for different types of WatDiv data on different query engines (gSMat$^+$, gSMat, gStore, RDF-3X) is shown in Figure \ref{fig:queryWatdiv}. The cyan lines in the figure represent the time trends of the gStore query, blue is RDF-3X, magenta is gSMat and red is gSMat$^+$.
As we can see, query results are analyzed on the same data but on different data scales.
Figures \ref{fig:queryWatdiv-F}, \ref{fig:queryWatdiv-L}, \ref{fig:queryWatdiv-S} and \ref{fig:queryWatdiv-C} show the trend of the query time of the F-type, L-type, S-type and C-type of watdiv data in different data scale.
With the continuous increase in the scale of data, various types of queries have also increased in time, and they have shown steady and gradual growth.
Obviously, it can be seen that gSMat's query timeline is lower than that of RDF-3X and gStore.
The gSMat$^+$ extended on the GPU is the best for query performance. Under any type of different data scale, the query time is always at the bottom.
Also, as the data size increases, the growth rate of the query time of gSMat and gSMat$^+$ is also the slowest.
From the graph, we can also analyze that the query time of gStore increases rapidly for each additional 100M of data, which is very easy to explain the data limit problem of gStore. Because the system memory it needed is very large.
And we can also get the fact that gSMat$^+$, which uses a GPU-accelerated join algorithm, performs better than gSMat when dealing with complex types.

\subsection{Experiments on Real Datasets}
\begin{table*}[!htbp]
	%\centering
	\caption{Runtime for YAGO queries(ms)}
	\label{tab:YAGO}
	\scriptsize
	\resizebox{\textwidth}{!}{%
		\begin{tabular}{|l|l|l|l|l|l|l|l|}
			\hline
			YAGO & Q1       & Q2     & Q3   & Q4   &  Q5  & Q6   &  Q7   \\
			\hline
			gSMat$^+$ & 10946.3 &	11598.5 &	15500.8 &	25713.9 &	24228.5 &	9923.3 &	21872.1  \\
			\hline
			gSMat & 18631.2 &	21116.0 &	18821.5 &	47026.4 &	40423.7 &	16895.3 &	48404.4  \\
			\hline
			RDF-3X   & 46885.3 &	615065.0 &	631522.0 &	208478.7 &	122193.7 &	65361.7 &	363859.7 \\
			\hline
			gStore   & 88371.0 &	327651.7 &	113275.7 &	898263.0 &	233928.7 &	783207.7 &	1342806.0  \\
			\hline
		\end{tabular}
	}
\end{table*}

\begin{table}[!htbp]
	%\centering
	\caption{Speedup for YAGO queries}
	\label{tab:speedupYAGO}
	\scriptsize
	\resizebox{\textwidth}{!}{%
		\begin{tabular}{|l|l|l|l|l|l|l|l|}
			\hline
			YAGO & Q1       & Q2     & Q3   & Q4   &  Q5  & Q6   &  Q7   \\		
			\hline
			gSMat/RDF-3X   & 2.5 &	29.1 &	33.6 &	4.4 &	3.0 &	3.9 &	7.5 \\
			\hline
			gSMat$^+$/RDF-3X   & 4.3 &	53.0 &	40.7 &	8.1 &	5.0 &	6.6 &	16.6  \\
			\hline
			gSMat/gStore   & 4.7 &	15.5 &	6.0 &	19.1 &	5.8 &	46.4 &	27.7 \\
			\hline
			gSMat$^+$/gStore   & 8.1 &	28.2 &	7.3 &	34.9 &	9.7 &	78.9 &	61.4 \\
			\hline
		\end{tabular}
	}
\end{table}

In this experiment, we test the performances of our method using the real datasets, YAGO and DBpedia.
Figure \ref{fig:queryrealdata} describes a comparison among the query runtime of real datasets YAGO and DBpedia on different engines (gSMat$^+$, gSMat, gStore, RDF-3X). The cyan lines in the figure represent the time trends of the gStore query, blue lines represent RDF-3X, magenta lines represent gSMat and red lines represent gSMat$^+$.
Next, we discuss query execution performance on both datasets separately.

First we analyze the comparison on YAGO data as shown in Figure~\ref{fig:YAGO}.
Due to YAGO does not provide benchmark queries, we design seven queries for YAGO real datasets. And it covers the four query types mentioned above.
As can be seen, the minimum query time is the extended gSMat$^+$ on GPU. This fully shows that our sparse matrix-based SPARQL join method has a very good scalability.
Even without GPUs, gSMat's query performance is better than gStore and RDF-3X.
The specific experimental results are shown in Table~\ref{tab:YAGO}~\ref{tab:speedupYAGO}.
As can be seen, compared with RDF-3X, the most obvious acceleration effect for gSMat is the Q3 (F-type), up to 33.6 times and for gSMat$^+$ is Q2 (C-type), up to 53 times. Overall, Q2 and Q3 are the most efficient comparisons of RDF-3X.
And compared with gStore, the  speedup for gSMat is the Q6 (S-type), up to 46.4 times and for gSMat$^+$ is also the Q6, up to 78.9 times. Overall, Q6 and Q7 are the most efficient comparisons of gStore.

Then we analyze the comparison on DBpedia data as shown in Figure~\ref{fig:DBpedia}.
For DBpedia dataset, we design four queries.
Analysis of query efficiency on DBpedia data is similar to YAGO.
As can be seen from the figure, the time of gSMat is smaller than gStore and RDF-3X.
For extended gSMat$^+$ on GPU, the query consumption time is reduced again.

Overall, for query efficiency analysis on real data, gSMat is significantly more efficient than RDF-3X and gStore.
Additional extensions on GPU, gSMat$^+$ improves performance again.

%Query runtime and join times for YAGO and DBpedia datasets
\begin{figure}[t]
	%**********************YAGO
	\subfigure[YAGO dataset]{\label{fig:yagoQuery}
		\begin{minipage}[htbp]{0.45\linewidth}
			\centering
			\scalebox{0.45}{
				\begin{tikzpicture}[font=\Large]
				\begin{semilogyaxis}[
				symbolic x coords={Q1,Q2,Q3,Q4,Q5,Q6,Q7},
				enlarge x limits=0.1,
				xtick = data,
				ybar,
				ymin=9500,ymax=5000000,
				legend style={at={(0.68, 0.98)},
					legend columns=2},
				bar width=2pt,
				ylabel={Query Runtimes (ms)},
				xlabel={Queries},
				ylabel near ticks,	
				]

				%gSMat_CPU
				\addplot
				[color=magenta, fill=magenta]
				coordinates {
					(Q1,18631.2) (Q2,21116.0) (Q3,18821.5) (Q4,47026.4) (Q5,40423.7) (Q6,16895.3) (Q7,48404.4)
				};

				%gSMat_GPU
				\addplot
				[color=red, fill=red]
				coordinates {
					(Q1,10946.3) (Q2,11598.5) (Q3,15500.8) (Q4,25713.9) (Q5,24228.5) (Q6,9923.3) (Q7,21872.1)
				};

				%gStore
				\addplot
				[color=cyan, fill=cyan]
				coordinates {
					(Q1,88371.0) (Q2,327651.7) (Q3,113275.7) (Q4,898263.0) (Q5,233928.7) (Q6,783207.7) (Q7,1342806.0)
				};
				
				%  RDF3X
				\addplot
				[color=blue, fill=blue]
				coordinates {
					(Q1,46885.3) (Q2,615065.0) (Q3,631522.0) (Q4,208478.7) (Q5,122193.7) (Q6,65361.7) (Q7,363859.7)
				};
				
				\legend{gSMat,gSMat$^+$,gStore,RDF-3X}
				\end{semilogyaxis}
				\end{tikzpicture}
			}
			\label{fig:YAGO}
		\end{minipage}
	}
	%**********************DBpedia
	\subfigure[DBpedia dataset]{\label{fig:DBpediaQuery}
		\begin{minipage}[htbp]{0.45\linewidth}
			\centering
			
			\scalebox{0.45}{
				\begin{tikzpicture}[font=\Large]
				\begin{semilogyaxis}[
				symbolic x coords={Q1,Q2,Q3,Q4},
				enlarge x limits=0.1,
				xtick = data,
				ybar,
				ymin=2000,ymax=4000000,
				legend style={at={(0.68, 0.98)},
					legend columns=2},
				bar width=2pt,
				ylabel={Query Runtimes (ms)},
				xlabel={Queries},
				ylabel near ticks,
				]
				
				%gSMat_CPU
				\addplot
				[color=magenta, fill=magenta]
				coordinates {
					(Q1,10753.9) (Q2,9922.5) (Q3,3379.3) (Q4,2054.8)
				};
				
				%gSMat_GPU
				\addplot
				[color=red, fill=red]
				coordinates {
					(Q1,9428.4) (Q2,5462.1) (Q3,2583.2) (Q4,2126.8)
				};
				
				%gStore
				\addplot
				[color=cyan, fill=cyan]
				coordinates {
					(Q1,94380.3) (Q2,744875.0) (Q3,576199.7) (Q4,3636.7)
				};
				
				%  RDF3X
				\addplot
				[color=blue, fill=blue]
				coordinates {
					(Q1,12178.7) (Q2,212971.3) (Q3,37399.3) (Q4,6330.7)
				};
				\legend{gSMat,gSMat$^+$,gStore,RDF-3X}
				\end{semilogyaxis}
				\end{tikzpicture}
			}
			\label{fig:DBpedia}
		\end{minipage}
	}
	\vspace*{-0.5cm}
	\caption{Average query runtime for YAGO and DBpedia (ms)}  \label{fig:queryrealdata}
\end{figure}
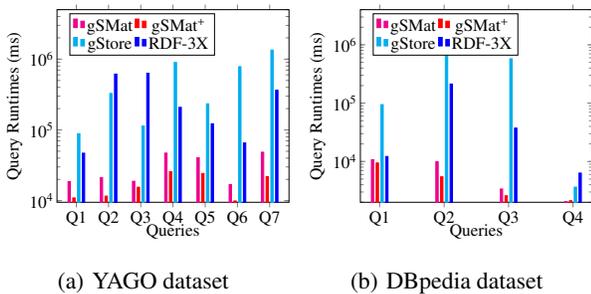

\section{Conclusions and future works}\label{sec:conclusion}
In this paper, we develop a sparse matrix-based SPARQL join algorithm and extend it in parallel using CUDA on GPU. And we propose a sparse matrix-based storage for exactly characterizing the sparsity of RDF graph data. It provides a compact and efficient method for storing RDF graph data physically and also supports high-performance algorithm of query execution. In addition, we analyze the overall cost by accumulating all intermediate results and then design a query plan generated algorithm.
We have analyzed and evaluated our proposal over both real and benchmark RDF datasets with half billion triples and demonstrated the high performance and scalability of our methods compared with the state-of-the-art RDF stores, i.e., RDF-3X and gStore. Our work on gSMat development continues along three dimensions. In the future, it is interesting to extend our query system to support more SPARQL basic operations, such as OPT, UNION, FILTER.
Moreover, it is valuable to further optimize gSMat$^+$ to take it more advantages in parallel query execution on GPU. Besides, we are going to bulid a distributed computing architecture to extend gSMat to support larger data queries.

\section*{Acknowledgments}
This work is supported by the National Key Research and Development Program of China (2016YFB1000603) and the National Natural Science Foundation of China (61672377,61502336).

\balance

%  ACKNOWLEDGMENTS are optional
%\vspace{-0.5em}
%This section is optional; it is a location for you
%to acknowledge grants, funding, editing assistance and
%what have you.  In the present case, for example, the
%authors would like to thank Gerald Murray of ACM for
%his help in codifying this \textit{Author's Guide}
%and the \textbf{.cls} and \textbf{.tex} files that it describes.
%

%\bibliographystyle{abbrv}
%\bibliography{sigproc}
%\bibliographystyle{abbrv}
%\bibliography{sigproc}  % sigproc.bib is the name of the Bibliography in this case

\end{document}